\documentclass[a4paper,USenglish,cleveref, autoref, thm-restate]{lipics-v2021}
\hideLIPIcs
\nolinenumbers
\usepackage[T1]{fontenc}
\usepackage{graphicx}
\usepackage[dvipsnames]{xcolor}
\usepackage{tikz}
\usepackage{xspace}
\usepackage{romanbar}
\usepackage{wrapfig}
\usepackage{complexity}
\usepackage{comment}

\usepackage{todonotes}

\newcommand{\TI}{\mathrm{I}}
\newcommand{\TII}{\mathrm{II}}
\newcommand{\TIII}{\mathrm{III}}
\newcommand{\TIV}{\mathrm{IV}}
\newcommand{\TV}{\mathrm{V}}
\newcommand{\TVI}{\mathrm{VI}}

\newcommand{\I}{{\Romanbar{I}\xspace}}
\newcommand{\II}{{\Romanbar{II}\xspace}}
\newcommand{\III}{{\Romanbar{III}\xspace}}
\newcommand{\IV}{{\Romanbar{IV}\xspace}}
\newcommand{\V}{{\Romanbar{V}\xspace}}
\newcommand{\VI}{{\Romanbar{VI}\xspace}}

\newcommand{\nI}{n_1\xspace}
\newcommand{\nII}{n_2\xspace}
\newcommand{\nIII}{n_3\xspace}
\newcommand{\nIV}{n_4\xspace}
\newcommand{\nV}{n_5\xspace}
\newcommand{\nVI}{n_6\xspace}

\newclass{\Q}{PseudoP}

\DeclareFontShape{T1}{lmr}{m}{scit}{<-> ssub * lmr/m/scsl}{}

\newtheorem{problem}{Problem}

\title{Perfect Rectangular Tilings with Two Colors} 

\author{Oswin Aichholzer}{TU Graz, Austria}{oswin.aichholzer@tugraz.at}{https://orcid.org/0000-0002-2364-0583}{}
\author{Robert Ganian}{TU Wien, Austria}{rganian@ac.tuwien.ac.at}{https://orcid.org/0000-0002-7762-8045}{Austrian Science Fund (FWF, Project 10.55776/Y1329) and Vienna Science and Technology Fund (WWTF, Project 10.47379/ICT22029).}
\author{Phillip Keldenich}{TU Braunschweig, Germany}{keldenich@ibr.cs.tu-bs.de}{https://orcid.org/0000-0002-6677-5090}{}
\author{Maarten L\"offler}{Utrecht University, the Netherlands}{m.loffler@uu.nl}{https://orcid.org/0009-0001-9403-8856}{}
\author{Gert Meijer}{Stenden University of Applied Sciences, the Netherlands}{gert.meijer@nhlstenden.com}{https://orcid.org/0009-0006-1184-2109}{}
\author{Ids de Vlas}{Utrecht University, the Netherlands}{i.l.devlas@students.uu.nl}{}{Austrian Science Fund (FWF, Project 10.55776/Y1329).}
\author{Alexandra Weinberger}{FernUniversit\"at in Hagen, Germany}{alexandra.weinberger@fernuni-hagen.de}{https://orcid.org/0000-0001-8553-6661}{}
\author{Carola Wenk}{Tulane University, United States}{cwenk@tulane.edu}{https://orcid.org/0000-0001-9275-5336}{National Science Foundation CCF 2107434}
\authorrunning{Aichholzer, Ganian, Keldenich, L\"offler, Meijer, de Vlas, Weinberger, Wenk}
\Copyright{Oswin Aichholzer and Robert Ganian and Phillip Keldenich and Maarten L\"offler and Gert Meijer and Ids de Vlas and Alexandra Weinberger and Carola Wenk} 
 %Please choose ACM 2012 classifications from https://dl.acm.org/ccs/ccs_flat.cfm 

 %TODO mandatory; please add comma-separated list of keywords

\category{} %optional, e.g. invited paper

\relatedversion{Full version of a paper with the same title at ISAAC 2026. A preliminary version of this work was presented as a poster at GD 2025~\cite{poster}.} %optional, e.g. full version hosted on arXiv, HAL, or other respository/website
\acknowledgements{This research was initiated at Dagstuhl Seminar 25201, ``Computational Geometry''. We thank the organizers and all participants for a stimulating atmosphere and valuable discussions. This research was made possible by the Centre for Unusual Collaborations (CUCo) under the project ``Going Off the Grid''.}
\begin{document}
% \tableofcontents
% \vspace*{5ex}

%\pI, \pII, \pIII, \pIV, \pV, \pVI\\
%\mypI

% \carola{The plan is to work on a submission to \underline{\href{https://www.algo-door.com/isaac2026/call-for-papers.html}{ISAAC}} for now. Deadline: June 29, 2026 (AoE). Length: 12 pages, excluding title page, references, abstract.}
% \newpage

% \section*{Brainstorming about Problem 1}

% {\bf Can we prove that our Problem 1, and the cleaner statement where we just sum unweighted $(a_i, b_i)$, are equivalent in the sense that a solution to one can be used to solve the other?}

% Comments and thoughts here...

\maketitle

\begin{abstract}
We study a finite tiling problem, where tiles are unit squares whose four edges are colored with one of two colors. We ask whether a given rectangle admits a
\emph{perfect rectangular tiling}: every cell of the rectangle is occupied by
one tile, neighboring edge colors match, and exactly
\(n_i\) tiles of type \(i\) are used, where rotations of the tiles are allowed.

Our problem is related to classical Wang tilings, more general finite
tile-placement problems, and edge placement puzzles. But in our problem, the
multiplicities of the tile types are part of the input and the tile alphabet is fixed and extremely small; thus the complexity of the problem arises from the
interaction between the rectangle dimensions and the prescribed tile
multiplicities.

We provide a comprehensive study of the perfect rectangular tiling problem. For this we consider all classes of subsets of the six possible tile types for two-colored edges, and we characterize for each class whether multiplicities either always allow a perfect rectangular tiling or whether their existence can be decided efficiently. 
\end{abstract}
\newpage

\setcounter{page}{1}

\section{Introduction}\label{sec:intro}

%\carola{Title brainstorming here}
%\newline
%Graph Tiles\\
% Tiling with a Bounded Number of Tiles\\
% Finite Tiling\\
% Tiling a Finite Region\\
% Finite graph tiling with two colors\\
% Tiling with Tile Quotas\\
%The Complexity of Two-Color Tiling with Tile Quotas\\
%The Complexity of Perfect Two-Color Tiling\\
% Perfect Two-Color Tiling(s)\\
% Perfect Two-Color Tiling(s) of Rectangles\\
% Perfect Rectangle Tiling with Two Colors\\
%Perfect Rectangular Tilings with Two Colors\\

Tiling problems ask how local compatibility rules constrain global
arrangements.  A tile system specifies which pieces may be placed next to one
another, and the central question is whether these local rules can be realized
on a prescribed domain.  In many classical tiling problems, each tile type is
available in unlimited supply.  In finite tiling problems, however, the
inventory itself is part of the input: one is given a finite number of tiles of
each type and asks whether a given region can be tiled using exactly those
tiles.

In this paper we study a particularly small but nontrivial finite tiling
problem.  The tiles are unit squares whose four edges are colored with one of
two colors. Adjacent tiles must have matching colors on their shared edge, and tiles may be rotated.  
We visualize these square tiles as having a single central vertex, with half-edges either present or absent toward each of the four sides.
 By allowing tiles to rotate, this gives us six distinct tiles, see Figure~\ref {fig:six_tiles}. 

\begin{figure}[!htb]
\centering
\includegraphics[page=2,scale=0.7]{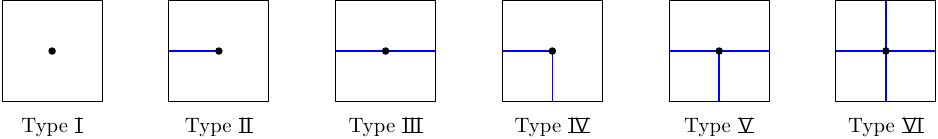}
\caption{The six different types of tiles: all possible tiles with at most two colors, up to rotation.}
\label{fig:six_tiles}
\end{figure}

\noindent
{\bf Problem statement.}
An instance of our problem consists of nonnegative integers $n_1,\ldots, n_6$,
specifying the number of available tiles of the six types $\I, \II, \ldots, \VI$, together with a
rectangle of size \(w\times h\).
We assume without loss of generality that $w \le h$, and we require that $wh=\sum_{i=1}^6 n_i$.  We ask whether the rectangle admits a
\emph{perfect rectangular tiling}: every cell of the rectangle is occupied by
one tile, rotations of tiles are allowed, neighboring edge colors match, and exactly
\(n_i\) tiles of type \(i\) are used.  
%Thus a necessary condition is $wh=\sum_{i=1}^6 n_i$.
The rectangle boundary does not impose additional constraints.
The problem is to decide whether this finite tile inventory can be arranged so that
the local matching constraints are satisfied everywhere.
See Figure~\ref {fig:cases} for some example tilings.

\begin{figure}[htbp]
  \centering
    \includegraphics [scale = 0.35] {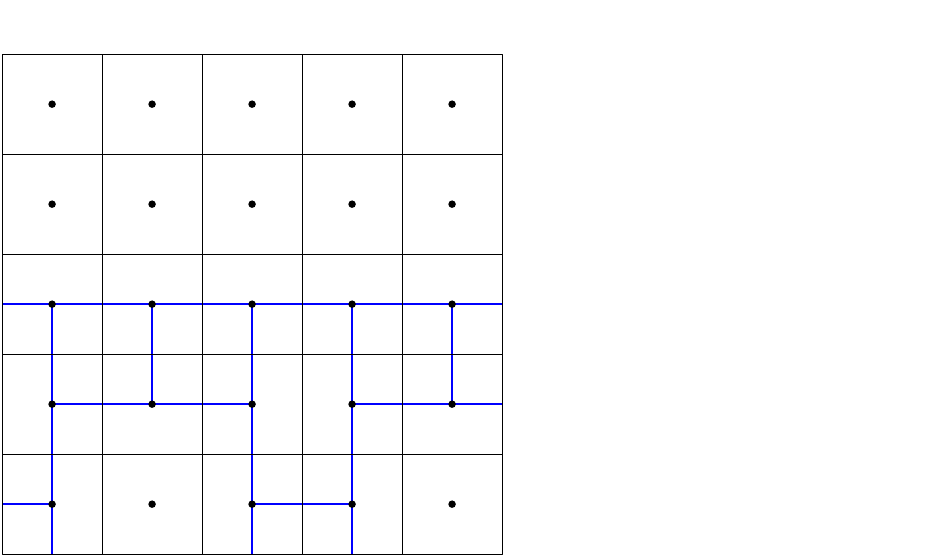}
    \hspace{-2.5cm}
    \includegraphics [scale = 0.35] {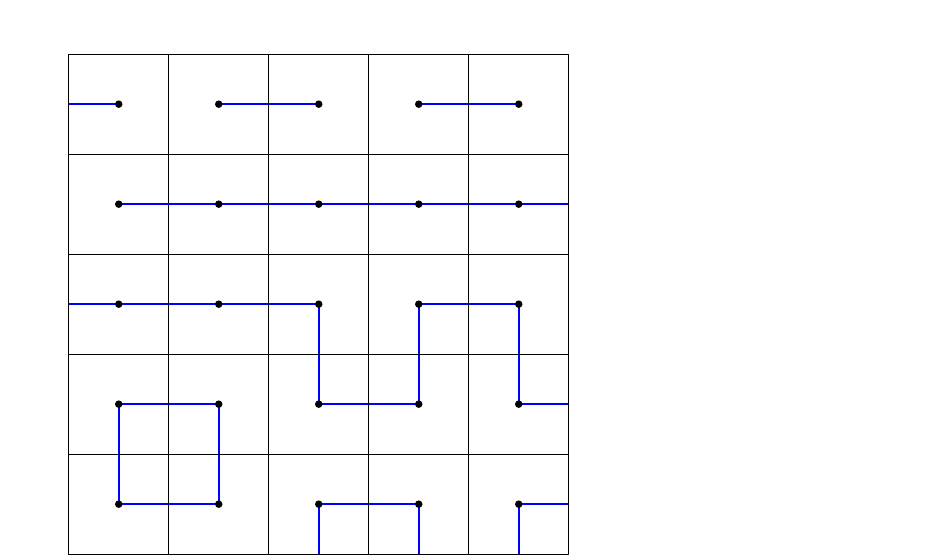}
    \hspace{-1.75cm}
    \includegraphics [scale = 0.35, page=2] {figures/I-III-VI}
%  \caption {Example tilings of a $5\times 5$ square with  different classes and numbers of available tile types: $\langle \I, \V \rangle$ with $n_1=12$ and $n_5=13$; $\langle \II, \III, \IV \rangle$ with $n_2=6$, $n_3=6$, and $n_4=13$; $\langle \I, \III, \VI \rangle$ with $n_1=6$, $n_3=13$, and $n_6=6$.}
\caption{Example tilings of a $5\times 5$ square. Left:  $\langle \I, \V \rangle$ with $n_1=12$, $n_5=13$. Center: $\langle \II, \III, \IV \rangle$ with $n_2=6$, $n_3=6$, $n_4=13$. Right: $\langle \I, \III, \VI \rangle$ with $n_1=6$, $n_3=13$, $n_6=6$.}
  \label {fig:cases}
\end {figure}

%The restriction to two colors gives a useful equivalent formulation.  Instead of placing tiles directly, one may assign a color \(0\) or \(1\) to every edge of the \(w\times h\) rectangular grid.  Each cell then sees four incident edge colors and therefore determines one of the six rotation classes above.  
%In this view, a perfect tiling is exactly a binary edge-coloring of the rectangular grid in which the six induced cell types occur with the prescribed multiplicities \(n_1,\ldots,n_6\).  The local matching condition is automatic, because adjacent cells share the same grid edge; the difficulty lies in realizing the prescribed global count vector.

This places our problem between classical Wang tiling and more general finite
tile-placement problems.  As in Wang tiling, local compatibility is expressed
by matching edge colors.  However, unlike the usual unlimited-supply setting,
the multiplicities of the tile types are part of the input.  At the same time,
the tile alphabet is fixed and extremely small: there are only two colors and
six rotation classes (tile types).  Thus the complexity of the problem cannot come from an
arbitrary collection of tile shapes or colors, but must arise from the
interaction between the rectangle dimensions and the six prescribed
multiplicities.

Such tiling problems also arise in applied domains such as material design.
In particular, the six tiles in Figure~\ref{fig:six_tiles} have been identified
as a minimal tile set rich enough to encode meaningful architected
materials~\cite{10.1002/adma.202305198}.  In virtual-growth approaches to
irregular architected materials, these tiles are assembled on a predefined grid
according to connectivity rules, and varying their relative concentrations can
produce distinct architectures with different mechanical and functional
properties~\cite{fox2024extractinggeometrytopologyorange,science2022,10.1002/adma.202305198}.

\subsection{Our Contributions}

We start by observing that the generalized problem of perfect tiling in a simple orthogonal polygon (rather than a rectangle) is \NP-hard even for orthogonally convex histograms; see Theorem~\ref {thm:hardness-simple-polygons} in the appendix.
However, when the domain is a rectangle, 
the input consists of a constant number of non-negative integers, and the encoding of these integers has a fundamental impact on the problem's complexity. In this paper we focus on the more general setting of binary encoding, which means that $w, h, $ and the $n_i$ can be exponential in the input size. %Thus a polynomial-time decision algorithm has to be polynomial in the sum of the logarithms 
In the computationally easier unary-encoded setting,
there are at most $n^{\mathcal{O}(1)}$ possible instances for each input size $n$.
Consequently, the unary-encoded variant of the problem belongs to the complexity class \P/$\mathsf{poly}$, and thus by the Karp-Lipton Theorem~\cite{KarpL80} cannot be \NP-hard unless the polynomial hierarchy collapses.

As it turns out, the perfect rectangular tiling problem in some cases requires solving the following number-theoretical problem or slight variants of it; see, for instance, \cref{thm:126} or \cref{thm:1246}.

\begin {problem}
  Given an integer $n$, find a sequence of tuples $(a_1, b_1), \ldots, (a_k, b_k)$ such that
  \begin {itemize}
    \item $\sum_{i=1}^k a_ib_i = n$, and
    \item $\sum_{i=1}^4 (a_i + b_i) + \sum_{i=5}^k (2a_i + b_i)$ is minimized.
  \end {itemize}
  \label{prob:mcd}
\end {problem}

Essentially, the most natural version of this problem (in which the weights of all $a_i$ and $b_i$ are uniform) asks for the decomposition of an integer area into integer rectangles that minimizes the total circumference.
The problem has ties to \emph{integer factorization}, a problem for which it is widely believed that no polynomial-time algorithm exists, although the complexity is unknown; it is also related to classical results such as Lagrange's four-square theorem.

%\noindent
{\bf Tile classes.}
Since the general problem is likely to be hard, we introduce \emph {classes}: a class groups instances in which not all tile types are present.
In this work, we provide a comprehensive study of the different classes in the perfect rectangular tiling problem. 
A tile type only contributes to a class if its multiplicity $n_i$ is non-zero. 
We therefore consider all classes of subsets of the six tile types that have non-zero multiplicities.
These classes also exhibit a {\em duality} property: if all sides with half-edges are replaced by sides without half-edges and vice versa, we obtain an equivalent dual problem: a tiling is valid iff its dual is.
Classes are denoted with $\langle . \rangle$ and duality with $\hat=$.

%\noindent
{\bf Contributions.}
For each class we present the status of the decision problem for this specific subset of non-zero tile multiplicities as summarized in Table~\ref{tab:results}. Here, YES means a perfect rectangular tiling always exists, for any choice of multiplicities. NO means such a tiling never exists. 
We cover these cases in Section~\ref{sec:yes_no}.
Then of course there are classes where for some values of multiplicities a perfect rectangular tiling exists, while for other values it does not. 
We provide algorithms for deciding whether such a tiling exists for a given input ($w, h, n_1, \ldots, n_6)$  in either polynomial or pseudo-polynomial time. 
Here, {\em polynomial time}, \P, denotes runtime polynomial in the input size  $\log(w)+\log(h)+\sum_{k=1}^6 \log(n_k)$. This is achieved by providing necessary and sufficient conditions on the values of the multiplicities and the rectangle dimensions, see Section~\ref{sec:P}.
{\em Pseudo-polynomial time}, \Q, denotes runtime polynomial in $w+h+\sum_{k=1}^6 n_k$. Carefully crafted dynamic programming algorithms provide these results, see Section~\ref{sec:pseudoP}.
While we settle most of the cases with either YES, NO, \P, or \Q, a few cases remain open, denoted with ``?'', where the subscript denotes a conjectured runtime. These are discussed in Section~\ref{sec:open}; most of the cases with label \Q{} are related to Problem~\ref{prob:mcd} or its variants.

% Each problem class may have value 
% YES (if all problems in the class are compatible with all rectangle dimensions), 
% NO (if none of the problems in the class are compatible), 
% \P{} (if compatible and incompatible problems both exist in the class, but it is possible to distinguish them in time polynomial in $\sum_{k=1}^6 \log(n_k)$), or
% \Q{} (if compatible and incompatible problems both exist in the class, but it is possible to distinguish them in pseudo-polynomial time, that is, time polynomial in $\sum_{k=1}^6 n_k$).
% Table~\ref{tab:results} summarizes our results, where remaining open cases are marked with ``?''.

\begin{table}[tbp]
\centering
\scriptsize
\setlength{\tabcolsep}{2pt}
\renewcommand{\arraystretch}{1.18}

\makebox[\linewidth][c]{%
\begin{tabular}{@{}l@{\hspace{0.45em}}r@{\hspace{0.45em}}|
                @{\hspace{0.45em}}l@{\hspace{0.45em}}r@{\hspace{0.45em}}|
                @{\hspace{0.45em}}l@{\hspace{0.45em}}r@{}}

\textbf{1 tile type} & &
\textbf{3 tile types} & &
\textbf{4 tile types} & \\[0.55em]

$\langle \I\rangle \hat= \langle \VI\rangle$ & YES &
$\langle \I,\II,\III\rangle \hat= \langle \III,\V,\VI\rangle$ & YES &
$\langle \I,\II,\III,\IV\rangle \hat= \langle \III,\IV,\V,\VI\rangle$ & YES \\

$\langle \II\rangle \hat= \langle \V\rangle$ & YES &
$\langle \I,\II,\IV\rangle \hat= \langle \IV,\V,\VI\rangle$ & YES &
$\langle \I,\II,\III,\V\rangle \hat= \langle \II,\III,\V,\VI\rangle$ & YES \\

$\langle \III\rangle$ & YES &
$\langle \I,\II,\V\rangle \hat= \langle \II,\V,\VI\rangle$ & YES &
$\langle \I,\II,\III,\VI\rangle \hat= \langle \I,\III,\V,\VI\rangle$ & $?_\Q$ \\

$\langle \IV\rangle$ & YES &
$\langle \I,\II,\VI\rangle \hat= \langle \I,\V,\VI\rangle$ & \Q &
$\langle \I,\II,\IV,\V\rangle \hat= \langle \II,\IV,\V,\VI\rangle$ & YES \\[0.55em]

\textbf{2 tile types} & &
$\langle \I,\III,\IV\rangle \hat= \langle \III,\IV,\VI\rangle$ & YES &
$\langle \I,\II,\IV,\VI\rangle \hat= \langle \I,\IV,\V,\VI\rangle$ & \Q \\

$\langle \I,\II\rangle \hat= \langle \V,\VI\rangle$ & YES &
$\langle \I,\III,\V\rangle \hat= \langle \II,\III,\VI\rangle$ & \P &
$\langle \I,\II,\V,\VI\rangle$ & \Q \\ %hopefully

$\langle \I,\III\rangle \hat= \langle \III,\VI\rangle$ & \P &
$\langle \I,\III,\VI\rangle$ & \P &
$\langle \I,\III,\IV,\V\rangle \hat= \langle \II,\III,\IV,\VI\rangle$ & YES \\

$\langle \I,\IV\rangle \hat= \langle \IV,\VI\rangle$ & YES &
$\langle \I,\IV,\V\rangle \hat= \langle \II,\IV,\VI\rangle$ & YES &
$\langle \I,\III,\IV,\VI\rangle$ & $?_\P$ \\ %The one with many cases

$\langle \I,\V\rangle \hat= \langle \II,\VI\rangle$ & \P &
$\langle \I,\IV,\VI\rangle$ & \P &
$\langle \II,\III,\IV,\V\rangle$ & YES \\[0.55em]

$\langle \I,\VI\rangle$ & NO &
$\langle \II,\III,\IV\rangle \hat= \langle \III,\IV,\V\rangle$ & YES &
\textbf{5 and 6 tile types} & \\[0.25em]

$\langle \II,\III\rangle \hat= \langle \III,\V\rangle$ & YES &
$\langle \II,\III,\V\rangle$ & YES &
$\langle \I,\II,\III,\IV,\V\rangle \hat= \langle \II,\III,\IV,\V,\VI\rangle$ & YES \\

$\langle \II,\IV\rangle \hat= \langle \IV,\V\rangle$ & YES &
$\langle \II,\IV,\V\rangle$ & YES &
$\langle \I,\II,\III,\IV,\VI\rangle \hat= \langle \I,\III,\IV,\V,\VI\rangle$ & $?_\Q$ \\

$\langle \II,\V\rangle$ & YES &
& &
$\langle \I,\II,\III,\V,\VI\rangle$ & $?_\Q$ \\

$\langle \III,\IV\rangle$ & YES &
& &
$\langle \I,\II,\IV,\V,\VI\rangle$ & $?_\P$\\ %Conjectured to be in P

& &
& &
$\langle \I,\II,\III,\IV,\V,\VI\rangle$ & $?_\Q$ \\

\end{tabular}
}

\caption{Overview of results for all classes of tile types.}
\label{tab:results}
\end{table}

\subsection{Related work}
Tiling problems have long been studied in mathematics and theoretical computer
science, from tilings of the plane to finite packing problems with local
compatibility constraints; see, e.g.,
\cite{tiling-book,tilings-and-patterns} for general background.
Wang tiles~\cite{6773658} are a central model in which square tiles carry
edge colors that must match on adjacent tiles, and tile self-assembly studies
related local-rule systems motivated by nanoscale
construction~\cite{demaine2008stagedselfassemblynanomanufacturearbitraryshapes}.
Our setting differs in that the tile alphabet is fixed, and the finite
multiplicities of the tile types are part of the input.

The closest work on the same local tile universe is due to Moore, Rapaport,
and Rémila~\cite{mrr-tgwt-02}, who study rotation-invariant tilings of finite
regions using two-color Wang tiles. Their setting has the same six rotation
classes, but assumes unlimited tile quantities and prescribed boundary colors.
Tyburec and Zeman~\cite{Tyburec_2023} study finite-domain Wang tiling in a
general integer-programming framework. Their model provides a broad bounded
formulation, but is not rotation-invariant and does not impose prescribed
finite multiplicities.

Finite edge-matching puzzles form another related line of work. Bosboom
et al.~\cite{Jigsaw1xn_JIP} study finite collections of square tiles with
colored edges that may be rotated and placed into a rectangular board. Their
formulation includes finite inventories, but their hardness and
inapproximability results are for unrestricted color alphabets; our problem is
a fixed-alphabet, two-color, six-type specialization. Ebbesen, Fischer, and
Witt~\cite{efw-empwr-11} study finite edge-matching puzzles under allowed
manipulations such as swaps and rotations. Their swaps-and-rotations variant
is closest to allowing translations and rotations of a finite tile inventory:
they give a linear-time algorithm for the single-row two-color case, while
their hardness results use larger color sets. Demaine and
Demaine~\cite{dd-jpemppcc-07} show equivalences among jigsaw puzzles,
edge-matching puzzles, and polyomino packing, proving NP-completeness for
general instances. These broader models allow richer tile or color alphabets,
but share the finite-multiset aspect that distinguishes our problem from
unlimited-supply tiling.

\section{YES and NO Classes\label{sec:yes_no}}
It is easy to see that all classes with one tile type always have a perfect rectangular tiling: one can always fill an arbitrary shape with repeated copies of the same tile (depending on the tile type, we may need to mirror alternating tiles horizontally or vertically).

%However, when multiple tile types have to be placed together, the problem becomes much more interesting. We start our investigation in this section by looking at classes with exactly {\em two} tile types, as understanding these will be key to the general problem. After covering all pairs of tile types, we state some general observations that will be useful later.

In only one case, $\langle \I, \VI \rangle$, the answer is always NO, since tiles of type $\I$ and $\VI$ can never be placed next to each other. In fact, the presence of both type $\I$ and type $\VI$  in a class will prove to be a good indicator for how complex the class is to deal with.

Of the remaining eight classes with two tile types, six of them always have a perfect rectangular tiling:   $\langle \I, \II \rangle \hat= \langle \V, \VI \rangle$,
  $\langle \I, \IV \rangle \hat= \langle \IV, \VI \rangle$,
  $\langle \II, \III \rangle \hat= \langle \III, \V \rangle$,
  $\langle \II, \IV \rangle \hat= \langle \IV, \V \rangle$,
  $\langle \II, \V \rangle$, and $\langle \III, \IV \rangle$.
This can be seen by filling the rectangle greedily in reading order (row-by-row top-down, left-to-right): First place all tiles of one type. Then place all tiles of the other type, while making minor adjustments to ensure minor parity constraints along the boundary are fulfilled; see \cref{fig:yes_cases}.

\begin{figure}[htbp]
    \centering
    \includegraphics[width=.9\textwidth]{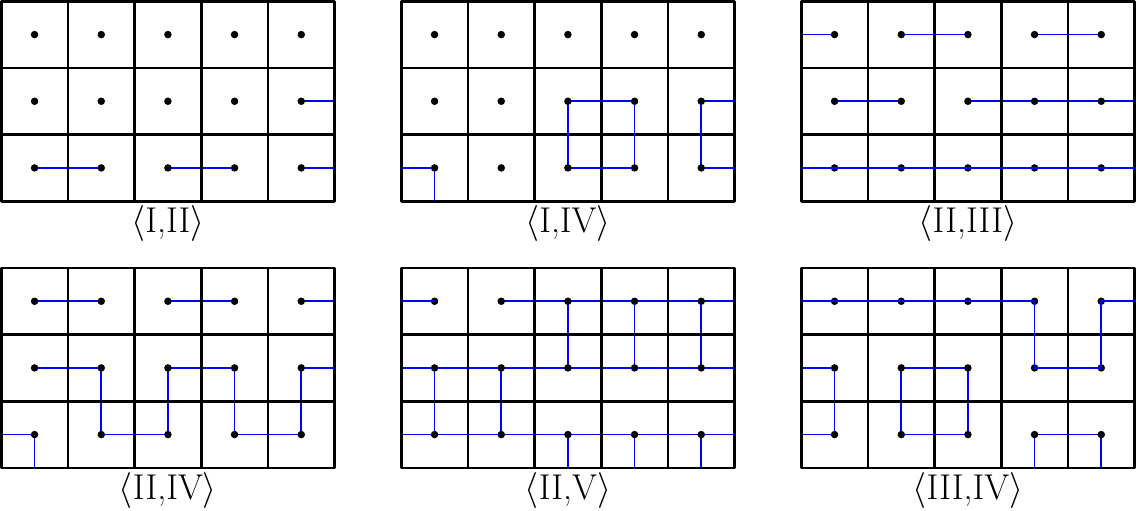}
    \caption{Greedy tilings for classes with two different tile types.
             The tiling for $\langle \I, \IV \rangle$ shows parity adjustments for the case $n_4 \equiv 3 \mod 4$.}
    \label{fig:yes_cases}
\end{figure}

The following lemma on filling rectangles with type $\IV$ tiles is useful for several cases.
\begin{lemma}
\label{lem:type4-constrained-rect}
Consider a $w' \times h'$-rectangle $R$ with constraints on up to two consecutive sides.
These constraints specify, for each cell edge on the constrained sides, whether it needs a halfedge or no halfedge to be pointing towards it from the interior of $R$.
Then, it is always possible to fill $R$ with $w'h'$ \IV-tiles.
\end{lemma}
\begin{proof}
Observe that we can independently choose whether the horizontal halfedge of a \IV-tile points left or right and whether the vertical halfedge points up or down.
W.l.o.g., let the top side (and optionally the left side) of $R$ carry constraints; otherwise, rotate or reflect $R$.
Then, we can construct a tiling of $R$ satisfying the constraints as follows:
when filling $R$ one cell at a time (in a column-by-column left-to-right, top-to-bottom fashion), we always have at most two constraints -- precisely one from the top and at most one from the left.
It is always possible to rotate a \IV-tile to satisfy such constraints; if we do not have a constraint from the left (in the leftmost column if the left side does not carry constraints), we arbitrarily choose to point the horizontal edge towards the left.
\end{proof}

For classes with three, four, or five tile types, there are also cases where the simple greedy approach guarantees the existence of a perfect rectangular tiling. We provide a proof sketch for classes with three tile types below. Additional proofs are in the appendix.
%We first investigate the classes for which a valid tiling always exist. Analogously to the classes in Theorem~\ref {thm:cases2}, we observe that whenever all instances of a class are realizable, they are realizable with a very simple method, which essentially places the tiles in lexicographical order with minimal deviations to repair edge cases.

\begin{restatable}{theorem}{thmyesThree}
  For classes
  $\langle \I, \II, \III \rangle \hat= \langle \III, \V, \VI \rangle$,
  $\langle \I, \II, \IV \rangle \hat= \langle \IV, \V, \VI \rangle$,
  $\langle \I, \II, \V \rangle \hat= \langle \II, \V, \VI \rangle$,
  $\langle \I, \III, \IV \rangle \hat= \langle \III, \IV, \VI \rangle$,
  $\langle \I, \IV, \V \rangle \hat= \langle \II, \IV, \VI \rangle$,
  $\langle \II, \III, \IV \rangle \hat= \langle \III, \IV, \V \rangle$,
  $\langle \II, \III, \V \rangle$, $\langle \II, \IV, \V \rangle$,
  a perfect rectangular tiling always exists.
\label{thm:cases3}
\end{restatable}

\begin{proof}[Proof Sketch]
We show the construction for class $\langle \TI, \TII, \TIV \rangle$, see also Figure~\ref{fig:sketch_I-II-IV}. Other classes follow a similar scheme. For class $\langle \TI, \TII, \TIV \rangle$, simply place the type $\I$ tiles first and continue with $\II$-tiles until they run out. Fill the remainder of a partially filled row with $\IV$-tiles, pointing their vertical halfedges downwards and alternating the orientation of their horizontal halfedges.
If unfilled rows remain, apply \cref{lem:type4-constrained-rect} to the remaining rectangle.
\end{proof}

\begin{figure}[htbp]
\centering
\includegraphics[width=.27\textwidth]{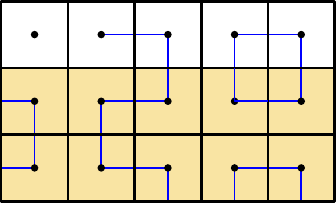}
\caption{Tiling constructed for class $\langle \TI, \TII, \TIV \rangle$. 
Yellow indicates 
%the yellow shaded area depicts an
application of \cref{lem:type4-constrained-rect}.}
\label{fig:sketch_I-II-IV}
\end{figure}

\begin{restatable}{theorem}{thmyesFourFive}\label{thm:cases4-5}
    For classes $\langle \TI, \TII, \TIII, \TIV \rangle \hat=\allowbreak \langle \TIII,\allowbreak \TIV,\allowbreak \TV,\allowbreak \TVI \rangle$, 
    $\langle \TI, \TII, \TIII, \TV \rangle \hat=\allowbreak \langle \TII, \TIII, \TV, \TVI \rangle$,
    $\langle \TI,\allowbreak \TII,\allowbreak \TIV,\allowbreak \TV \rangle \hat=\allowbreak \langle \TII,\allowbreak \TIV,\allowbreak \TV,\allowbreak \TVI \rangle$, 
    $\langle \TI,\allowbreak \TIII,\allowbreak \TIV,\allowbreak \TV \rangle \hat= \langle \TII, \TIII, \TIV, \TVI \rangle$,
    $\langle \TII, \TIII, \TIV, \TV \rangle$, and
    $\langle \TI,\allowbreak \TII,\allowbreak \TIII,\allowbreak \TIV,\allowbreak \TV \rangle \hat= \langle \TII,\allowbreak \TIII,\allowbreak \TIV,\allowbreak \TV,\allowbreak \TVI \rangle$, a perfect rectangular tiling always exists.
\end{restatable}

\section {Polynomial Classes\label{sec:P}}

\subsection{\texorpdfstring{$\langle \TI, \TIII \rangle\hat= \langle \TIII, \TVI\rangle$ and $\langle \TI, \TV \rangle\hat= \langle \TII, \TVI\rangle$ }
{<I,III>=<III,VI> and <I,V>=<II,VI>}}

%For a full classification of classes with two tile types, this leaves the two cases $\langle \I, \III \rangle$ and $\langle \I, \V \rangle$, which can be handled in polynomial time as follows.

\begin {theorem}
  For 
  $\langle \I, \III \rangle \hat= \langle \III, \VI \rangle$ a perfect rectangular tiling exists iff $\nIII \mid w$ or $\nIII \mid h$.
For $\langle\I, \V\rangle \hat= \langle \II, \VI \rangle$ a perfect rectangular tiling exists iff $\nV \in \{w, h, w+h-1\}$ or $\nV \geq 2w$.
\end {theorem}
\begin{proof}
If $\nIII$ divides $w$ or $h$, we can construct a tiling by making full rows or columns of \III-tiles.
On the other hand, if a row contains a \III-tile such that its left and right boundaries meet a halfedge, then all tiles in that row must be \III-tiles in the same orientation.
A symmetric argument holds for columns, implying that a tiling is impossible if $\nIII$ does not divide $w$ or $h$.
This implies a polynomial-time algorithm for deciding whether a tiling exists.
%\end{proof}

% \begin{theorem}
%     For class $\langle\I, \V\rangle \hat= \langle \II, \VI \rangle$ a tiling of a $w \times h$ rectangle where $w \leq h$ is possible if and only if $\nV \in \{w, h, w+h-1\}$ or $\nV \geq 2w$.
% \end{theorem}
%\begin{proof}
For $\langle\I, \V\rangle \hat= \langle \II, \VI \rangle$, to see that a tiling does not exist for $0 < \nV < w$, we observe the following. 
If there are fewer than $w$ tiles, there must be at least one tile of type $\V$ that is adjacent to two tiles that are not of type $\V$ and thus of type $\I$.
However, tiles of type $\V$ have only one side that matches to tiles of type $\I$.
The same holds if we have $w < \nV < \min(h, 2w)$ or $h < \nV < \min(2w, w+h-1)$, thus a tiling is impossible in these cases.

For the other direction, exactly $w$ tiles can be put in a boundary row and exactly $h$ tiles can be put in a boundary column, producing valid tilings.
Exactly $w+h-1$ tiles can be put on two boundary sides, sharing a corner; see \cref{fig:type-1-5}.
For any $\nV \geq 2w$, we first place full rows of $\I$-tiles.
We can then place 
%all 
$\V$-tiles in a single block of complete rows,
possibly ending with an incomplete row of $\V$-tiles to incorporate the remaining $\I$ tiles; see \cref{fig:type-1-5}~(right) and note that we can extend the block of $\V$-tiles by complete 
rows of $\V$-tiles.
\begin{figure}[hbtp]
    \centering
    \includegraphics[scale=0.67]{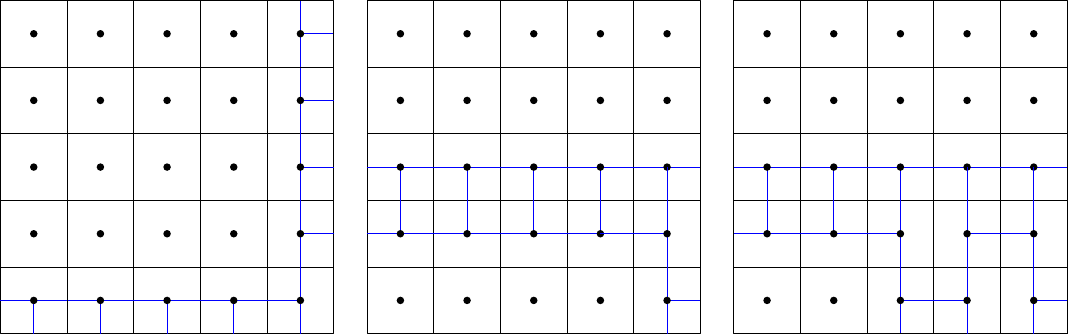}
    \caption{Realization of $w+h-1$, $2w+1$ and $2w+3$ tiles of type \V,
             showing how incomplete rows of type \V-tiles can be achieved as soon as $\nV \geq 2\min(w,h)$.}
    \label{fig:type-1-5}
\end{figure}
\end{proof}

%\subsection{P cases with 3 tile types}

%On top of the eight classes with value YES, we obtain three classes with value \P.

\subsection{\texorpdfstring{$\langle \TI, \TIII, \TV \rangle \hat= \langle \TII, \TIII, \TVI \rangle$}
{<I,III,V>=<II,III,VI>}}

\begin{restatable}{theorem}{thmIiiiV}
%  Let $R$ be a $w \times h$-rectangle with $w \leq h$, and let $\nI, \nIII, \nV > 0$ and $\nII, \nIV, \nVI = 0$.
  Let $\nIII = kw + r$ for $k \in \mathbb{N}_0$ and $r < w$ and $\nV = qh + s$ for $q \in \mathbb{N}_0$ and $s < h$.
  Then, a perfect rectangular tiling of class $\langle \TI, \TIII, \TV \rangle $ exists iff $\nIII + \nV \geq w$ and one of the following conditions hold: (1)~$w \mid (\nIII + \nV)$, (2)~$h \mid (\nIII + \nV)$, (3)~$\nV \geq \max(2, 2(w-\nIII))$, or (4)~$\nV = 1 \wedge (r < h - k \vee s < w - q)$.
\label{thm:i-iii-v-is-p}
\end{restatable}

The proof of \cref{thm:i-iii-v-is-p} follows the cases, and is available in the appendix.

\subsection{\texorpdfstring{$\langle \TI, \TIII, \TVI \rangle$}
{<I,III,VI>}}
\begin{theorem}
A perfect rectangular tiling of class $\langle \TI, \TIII, \TVI \rangle $ exists iff
there are integers $0 < a \leq w$ and $0 < b \leq h$ such that $\nI = wh - ah - bw + ab$, $\nIII = ah + bw - 2ab$, and $\nVI = ab$. The existence of such $a$ and $b$ can be checked in polynomial time.
\end{theorem}

\begin{proof}
Assume we have a tiling. A $\VI$-tile forces propagation in both directions:
its row must consist only of $\VI$-tiles and horizontally oriented
$\III$-tiles, while its column must consist only of $\VI$-tiles and vertically
oriented $\III$-tiles. Similarly, any horizontally, respectively vertically,
oriented $\III$-tile forces the same structure along its entire row,
respectively column. Thus every occurrence of a $\VI$-tile or an oriented
$\III$-tile induces a straight line segment spanning the corresponding row or
column.

Thus, in the tiling, we end up with $a$ columns and $b$ rows for $a < w, b < h$ that have such a straight line through the whole row or column, and the crossings between the horizontal straight lines and vertical straight lines are exactly the $\VI$-tiles.
In those $a$ columns and $b$ rows, there are $ah+bw-ab$ tiles, exactly $ab$ of which are of type $\VI$.
Thus $ah + bw - 2ab$ are of type $\III$.
The remaining $wh - ah - bw + ab$ are $\I$-tiles; we remark that $a,b>0$ since there is at least one $\VI$-tile.
Thus, if the multiplicities cannot be represented this way with integers $a < w$ and $b < h$, there cannot be a tiling.
If they can be represented  this way, an arbitrary arrangement of $b$ filled rows and $a$ filled columns produces a valid tiling.

Determining whether such $a$ and $b$ exist is simple:
With $n_6=ab$ we have $n_3=ah+bw-2n_6$, which yields $b=(n_3-ah+2n_6)/w$. 
Given $w, h, \nI, \nIII, \nVI$, we can thus compute the candidates for $b$ and $a$, and check their integrality, whether they correctly produce $n_1=wh-ah-bw+ab$, and whether they are within our bounds.

% by $\nI = wh - ah - bw + \nVI \Leftrightarrow a = (- b w + h w - n_{1} + n_{6})/h$ and $\nVI = ab$, we find two solutions for $b$:
% \[b = \frac{h w - n_{1} + n_{6} \pm \sqrt{h^{2} w^{2} - 2 h n_{1} w - 2 h n_{6} w + n_{1}^{2} - 2 n_{1} n_{6} + n_{6}^{2}}}{2 w}.\]
% Given $w, h, \nI, \nIII, \nVI$, we can thus compute the candidates for $b$ and $a$ if they are real, check their integrality, whether they correctly produce the given $\nIII$ and whether they are within our bounds.
\end{proof}

%\vspace{1em} \noindent {\bf value:} \P{}\\

\subsection{\texorpdfstring{$\langle \TI, \TIV, \TVI \rangle$}
{<I,IV,VI>}}
%The following theorem allows us to decide compatibility in polynomial time for this class.
\begin{theorem}
A perfect rectangular tiling of class $\langle \TI, \TIV, \TVI \rangle $ exists
%    If $\nI, \nIV, \nVI > 0$ and $\nII,\nIII,\nV = 0$, a valid tiling is possible 
    iff $\nIV \geq w$ or $\min\{\nI,\nVI\} \leq \nIV(\nIV-1)/2$.
\end{theorem}
\begin{proof}
Tiles of type $\I$ and $\VI$ can never be directly adjacent, so they must be separated by a layer of $\IV$-tiles, which
%This separating layer 
must be at least one tile ``thick'';
see \cref{fig:i-iv-vi}.
Also recall that $w \leq h$.

\begin{figure}[htbp]
    \centering
    \includegraphics[width=.6\textwidth]{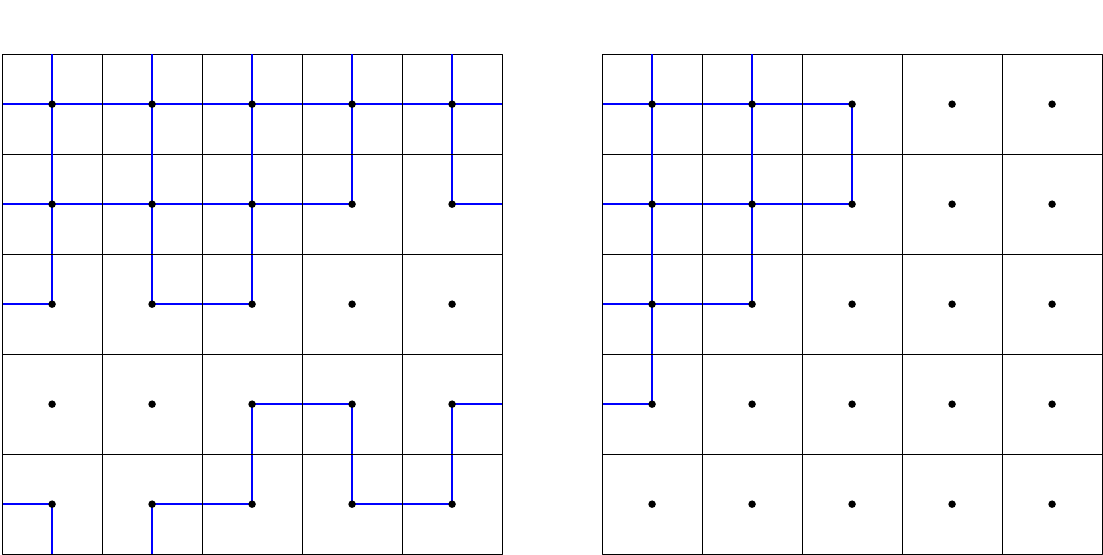}
    \caption{Left: placement of tiles if $\nIV \geq w$. Right: separating a triangular region 
    %(one tile less than the maximum achievable) 
    by $\nIV$-tiles.}
    \label{fig:i-iv-vi}
\end{figure}

If $\nIV \geq w$, we place the tiles as follows.
We first place all $\VI$-tiles in reading order until they run out.
Then, we place exactly $w$ $\IV$-tiles, with their half-edges always connecting to the $\VI$-tiles.
Next, we place all $\I$-tiles until they run out, fill any partially filled row with \IV-tiles and apply \cref{lem:type4-constrained-rect} to the remaining rows.

If $\nIV < w$, we can only separate a small triangular region, e.g., in the top-left corner; this region can either contain the $\VI$-tiles or the $\I$-tiles.
W.l.o.g., assume $\nVI \leq \nI$, so the triangular region should contain the \VI-tiles.
Given $\nIV$ \IV-tiles, we can separate a triangular area of at most $A = (\nIV-1) + (\nIV - 2) + \cdots = \frac{\nIV(\nIV-1)}{2}$ tiles.
If $\nVI \leq A$, we can achieve a tiling; if $\nVI < A$, we can shrink the separated area as necessary by shortening its longest rows first.
No other separation scheme allows us to separate more area given the same number of $\IV$-tiles; thus, if $\nVI > A$, no tiling is possible.
\end{proof}

%\subsection{P cases with four tile types}

\section {Pseudopolynomial Classes}\label{sec:pseudoP}

Interestingly, in some cases, no polynomial-time algorithm appears to exist, because tiles must be clustered in specific patterns. In these cases, we present pseudo-polynomial-time algorithms and argue that this may be necessary due to connections to integer factorization.
We focus on classes
$\langle \I, \II, \VI \rangle \hat= \langle \I ,\V, \VI \rangle$ and 
$\langle \I, \II, \IV, \VI \rangle \hat= \langle \I, \IV, \V, \VI \rangle$ below. For full proof details as well as case $\langle \I, \II, \V, \VI \rangle$ please see the appendix.

\subsection{\texorpdfstring{$\langle \TI, \TII, \TVI \rangle \hat= \langle \TI ,\TV, \TVI \rangle$}{<I,II,VI>=<I,V,VI>}} 
%The only remaining case with three tile types is $\langle \TI, \TII, \TVI \rangle \hat= \langle \TI ,\TV, \TVI \rangle$.
This case is related to interesting and natural number-theoretic questions,
which are, to the best of our knowledge, still open.
The essential observation is that $\I$- and $\VI$-tiles must not be adjacent and must thus be separated by a \emph{frame} of $\II$-tiles.
Simultaneously, $\II$-tiles can only separate $\VI$-tiles from $\I$-tiles if the $\VI$-tiles form a rectangle: if a tile is adjacent to $\VI$-tiles on more than one side, it must itself contain a $\VI$-tile and thus, no shape formed by $\VI$-tiles can have a reflex corner.
We begin our discussion 
%of this case 
by a few sufficient conditions that guarantee that a tiling is possible.
\begin{lemma}
    %Let $n = \min\{w,h\}$.
    %, $\nI, \nII, \nVI > 0$ and $\nIII, \nIV, \nV = 0$.
    Then, $\nII \geq 2w$ guarantees that a perfect rectangular tiling exists for class $\langle \TI, \TII, \TVI \rangle$.
    \label{lem:2n-2-tiles-suffice}
\end{lemma}
\begin{proof}
   Recall that $w \leq h$.
   To obtain a tiling, we build full rows of $\VI$-tiles at the bottom of our rectangle until fewer than $w$ tiles remain, separating the full rows by a row of $w$ $\II$-tiles.
   At most $w-1$ $\VI$-tiles remain, which can be placed in the topmost row and surrounded by at most %$(w-1) + 1$ 
   $w$ $\II$-tiles to separate them.
   $\nII \geq 2w$ guarantees that at least two full rows can be filled with $\II$-tiles, ensuring there will be no overlap between the $\II$-tiles used to separate the full $\VI$-tile rows at the bottom and the $\II$-tiles separating the partially filled row at the top.
\end{proof}

Furthermore, for any fixed $w, h$ and $\nVI$, replacing $\I$-tiles by $\II$-tiles never hurts. See the appendix for the proof.
%\begin{lemma}
\begin{restatable}{lemma}{lemmintwoopt}
Let $w, h, \nI, \nII, \nVI > 0$ such that a perfect rectangular tiling exists for class $\langle \TI, \TII, \TVI \rangle$.
    Let $\nII' \geq \nII$ and $\nI' = \nI - (\nII'-\nII) > 0$.
    Then there also exists a perfect rectangular tiling for $w, h, \nI', \nII', \nVI$.
    \label{lem:min-2-is-optimal}
\end{restatable}

Therefore, we can rephrase the perfect rectangular tiling problem as follows:
given $w, h$ and $\nVI$, subdivide $\nVI$ into integer rectangles $\nVI = \sum_i a_i b_i$ with $0 < a_i \leq w$, $0 < b_i \leq h$ such that the exposed circumference, i.e., the number $\nII$ of necessary $\II$-tiles, is minimized.
The exposed circumference of each rectangle is $w$ if $a_i = w$ (full rows),
$h$ if $b_i = h$ (full columns), $a_i + b_i$ for other rectangles placed in a corner, $a_i + 2b_i$ or $2a_i + b_i$ for other rectangles placed on the boundary and $2a_i + 2b_i$ otherwise.
We make the following observation.
\begin{lemma}
    Let $\nVI = kw + r$, $\nVI = \ell h + s$ with $k,\ell \in \mathbb{N}_0$, $r < w$ and $s < h$.
    Then we need at least 
    $\nII \geq \min\left\{
        w + 2\left\lceil\sqrt{r}\right\rceil,
        h + 2\left\lceil\sqrt{s}\right\rceil,
        2\left\lceil\sqrt{n_6}\right\rceil
    \right\}$
    tiles of type $\II$ in any perfect rectangular tiling for class $\langle \TI, \TII, \TVI \rangle$.
    \label{lem:valid-lower-bounds}
\end{lemma}
\begin{proof}
    We consider three mutually exclusive and exhaustive cases: (1)~some of the $\VI$-tiles fill complete rows,
    (2)~some of the $\VI$-tiles fill complete columns and (3)~there is no complete column or row containing a $\VI$-tile.
    In case (1), even assuming an arbitrary number of corners being available instead of just $2$, by the arithmetic-geometric mean inequality and the  concavity of $\sqrt{\,\cdot\,}$, we need at least $2\lceil\sqrt{r}\rceil$ $\II$-tiles to handle the $\VI$-tiles that do not fit into full rows; an analogous argument works for $s$ in case (2) and all $n_6$ tiles in case (3).
\end{proof}

This lemma yields a lower bound of $2\lceil\sqrt{\nVI}\rceil$ if $\nII < \min\{w,h\}$ since options (1) and (2) are unavailable in this case.
On the other hand, we have the following sufficient condition.
\begin{lemma}
%    For any $w \times h$-rectangle and any $\nI, \nII, \nVI > 0$ with $\nIII, \nIV, \nV = 0$ and $\nI + \nII + \nVI = wh$,
If $\nII \geq 4\sqrt{\nVI}$, then there exists a perfect rectangular tiling for class $\langle \TI, \TII, \TVI \rangle$.
    \label{lem:valid-upper-bounds}
\end{lemma}
\begin{proof}
Recall that $w \leq h$.
    We analyze two cases, (a)~$\nVI < w^2/4$ and (b)~$\nVI \geq w^2/4$.

    First consider case (a).
    Lagrange's four-square theorem states that every integer $\nVI$ can be written as sum of four non-negative integer squares $\nVI = a_1^2 + \cdots + a_4^2$.
    We can place each square into a corner of our region if $a_i + a_j \leq w-2$ for each pair $i \neq j$, avoiding overlaps between the frames of $\II$-tiles surrounding each square.
    For any pair $i \neq j$, we have $a_i^2 + a_j^2 \leq \nVI < w^2/4$.
    By the Cauchy-Schwarz inequality, we obtain
    \[(a_i + a_j)^2 \leq 2(a_i^2 + a_j^2) < w^2/2 \implies a_i + a_j < w/\sqrt{2},\]
    which guarantees $a_i + a_j \leq w-2$ as soon as $w/\sqrt{2} \leq w - 1$, which holds for all $w \geq 4$.
    For $w \leq 3$, due to $\nVI < w^2/4 \leq 9/4$, we only need to consider $\nVI \in \{1, 2\}$; for $\nVI = 1$, $\nII = 2 \leq 4$ suffices, and for $\nVI = 2$, $\nII \leq 3 < 4\sqrt{2}$ suffices, so a valid tiling exists in these small cases.
    
    For $w \geq 4$, we can place each square in its own corner without overlap.
    The exposed circumference of these squares is $2\sum_{i=1}^4 a_i$.
    Using Cauchy-Schwarz again, we obtain \[\left(\sum_{i=1}^4 a_i\right)^2 \leq 4\sum_{i=1}^4a_i^2 = 4\nVI \implies 2\sum_{i=1}^4 a_i \leq 4\sqrt{\nVI},\]
    proving our bound in case (a).
    For case (b)~$\nVI \geq w^2/4$, we have $\sqrt{\nVI} \geq w/2$, so $4\sqrt{n_6} \geq 2w$; the claim follows directly from \cref{lem:2n-2-tiles-suffice}.
\end{proof}

We use these lemmas to prove the following theorem.
\begin{theorem}
For class $\langle \TI, \TII, \TVI \rangle$ there is a pseudo-polynomial algorithm that decides whether a perfect rectangular tiling exists.
\label{thm:126}
\end{theorem}
\begin{proof}
    By \cref{lem:min-2-is-optimal}, we know that, given $n_6, w, h$, it suffices to find the minimum $n_2^*$ that admits a tiling to decide whether a given $n_2$ admits a tiling ($n_2 \geq n_2^*$) or not.
    Furthermore, 
    %w.l.o.g., assuming $w \leq h$ 
    we know by \cref{lem:2n-2-tiles-suffice} that $n_2 \geq 2w$ already implies a tiling.
    We can use the following algorithm to find $n_2^*$.
    We consider three cases: (1)~filling complete columns,
    (2)~complete rows, or (3)~neither complete rows nor columns.
    
    We first describe case (1); case (2) is symmetric.
    There are only polynomially many choices for the number of complete columns we fill, so we can try each choice $c$ individually.
    For each $c$, we have to find a decomposition of the remaining number of $\VI$-tiles $n_6' = n_6 - ch$ into rectangles that minimizes the number of $\II$-tiles needed.
    We can assume these rectangles to be placed with their longer side against the boundary opposite the block of complete columns, w.l.o.g. the right boundary, since that minimizes the number of $\II$-tiles needed.
    Note that we cannot run out of space for the optimal choice $c$ and the optimal decomposition: if the only way of optimally decomposing $n_6'$ requires a packing of rectangles that cannot be turned into one where all rectangles touch the right boundary with their longer side by rotating and moving rectangles, that is due to an overlap in vertical direction, which implies that each row, on average, contains at least $2$ $\II$-tiles and we thus have $n_2 \geq 2w$.
    
    Thus, as long as we can guarantee that the shorter side of each rectangle is bounded by $B_m = w - c - 2$ (such that it will not overlap with the complete columns), our goal is to subdivide the remaining number $n_6'$ of $\VI$-tiles into rectangles $a_1 \times b_1, \ldots, a_k \times b_k$ with $a_i \geq b_i$, $b_i \leq B_m$ and $b_i \leq b_{i-1}$ for all $i$ such that $a_1 + b_1 + a_2 + b_2 + \sum_{i=3}^k a_i + 2b_i$ is minimized; note that placing the first two rectangles (those with the largest $b_i$) in the corner minimizes the number of $\II$-tiles needed.

    We can solve this problem via dynamic programming (keeping the value of $B_m$ fixed in the DP).
    For any value $c$, our DP has a table $D[x,k]$ with three entries for each value $x \leq n_6'$, each entry corresponding to the minimum number of $\II$-tiles needed for $x$ $\VI$-tiles if $k \in \{0, 1, 2\}$ corners are still available; our answer for $c$ is $D[n_6',2]$.
    We can initialize $D[0,k] = 0$ and compute each cell by considering each potential pair of side lengths of a rectangle with at most $x$ tiles, taking the minimum over all results and combining with previously computed entries to handle the remaining tiles.

    Case (3) can be handled in a very similar manner, using a DP table with $k \leq 4$ available corners and $B_m = w$.
    We once again sort our rectangles in a decreasing order such that $b_i \leq b_{i-1}$.
    We then reorder the first four rectangles $r_1,...,r_4$ of size $a'_1 \times b'_1,...,a'_4 \times b'_4$ by their longer side, i.e., such that $a'_i \leq a'_{i-1}$.
    
    First we try to place $r_1$ in the top-left corner and $r_2$ in the bottom right corner; both are placed vertically (portrait). If these two rectangles overlap then $a'_1 + a'_2 \geq w \land b'_1 + b'_2 \geq w$, thus we would need at least $2w$ $\II$-tiles.
    Next we attempt to place $r_3$ in the top right corner, once again vertically.
    If this causes an overlap with $r_1$ it must be with their shortest side, thus once again causing us to need at least $2w$ $\II$-tiles.
    If this causes an overlap with $r_2$ we place the third rectangle horizontally instead, if this placement still leads to an overlap with $r_2$, we must use at least $2w$ $\II$-tiles again since $a'_1 > a'_2$ and $a'_3 > b'_3$.
    If it leads to an overlap with $r_1$ a similar argument applies since $a'_1 > b'_1$ and $a'_2 > a'_3$.
    If $r_3$ is placed horizontally without overlaps, we place $r_4$ in the remaining corner horizontally as well, without overlaps because of a symmetrical argument to the one above.
    Any remaining rectangles are placed between $r_2$ and $r_4$, if this causes an overlap then it must be the case that $b'_2 + a'_4 + \sum_{i=5}^k a_i + 2b_i \geq w-2$. However, since we know $a'_2 + a'_3 \geq w-1$, combined with the fact that $a'_1,b'_1$ and $b'_3$ are all at least 1, the number of needed $\II$-tiles is once again at least $2w$.
    If $r_3$ is placed horizontally without overlaps but $r_4$ causes overlap instead, a symmetrical argument can be made to show the need for at least $2w$ $\II$-tiles.
    % this looks ok now; should we submit another copy?
    % Yes I think we can, shall I do it?
    
    If $r_1,...,r_4$ are all placed horizontally without overlap, the remaining rectangles can be placed with their longest side on the top. If that does not fit they get placed at the bottom instead, if that does not fit either, we once again have an average of at least two $\II$-tiles per column and thus $2w$ total $\II$-tiles.

\end{proof}

This class is at least as hard as Problem \ref{prob:mcd}.

\subsection{\texorpdfstring{$\langle \TI, \TII, \TIV, \TVI \rangle \hat= \langle \TI, \TIV, \TV, \TVI \rangle$}
{<I,II,IV,VI>=<I,IV,V,TVI>}}

\begin{restatable}{theorem}{ppAndRSAHardness}
    For the class $\langle \TI, \TII, \TIV, \TVI \rangle \hat= \langle \TI, \TIV, \TV, \TVI \rangle$, there exists a pseudo-polynomial time algorithm deciding whether a tiling exists.
    \label{thm:pp-and-rsa-hardness}
    \label{thm:1246}
\end{restatable}

\begin{proof}[Proof Sketch]
We have to separate the $\I$-tiles and $\VI$-tiles by $\II$- and $\IV$-tiles, since $\I$- and $\VI$-tiles cannot be adjacent.
We distinguish the following cases, depicted in \cref{fig:1246}:
(1)~a horizontal or vertical path consisting of $\II$- and $\IV$-tiles connecting two opposite boundaries separates two blocks of $\I$- and $\VI$-tiles,
(2)~all $\VI$-tiles are placed in a region of the grid that contains one corner, separated by a path of $\II$- and $\IV$-tiles connecting two consecutive boundaries,
(3)~all $\I$-tiles are placed in a region of the grid that contains one corner, separated by a path of $\II$- and $\IV$-tiles connecting two consecutive boundaries.
We argue in the full proof in the appendix why 
%the remaining cases, namely multiple disconnected regions of $\VI$- or $\I$-tiles separated by closed loops of $\II$- and $\IV$-tiles or regions on the boundary separated by a path of $\II$- or $\IV$-tiles can be transformed into tilings that fall into one of the cases (1)--(3) and thus, we can assume that 
this case distinction is exhaustive.
% The idea of the argument for this is that, if we have sufficiently many $\II$- and $\IV$-tiles, we can compute tilings in which we pretend the excess $\II$- and $\IV$-tiles to be $\I$-tiles and then repair these by replacing some $\I$-tiles with $\II$- and $\IV$-tiles.
% Intuitively speaking, this allows us to focus on minimizing the surface area where the regions containing $\I$-tiles and $\VI$-tiles touch, which is achieved by one of the cases (1)--(3).

\begin{figure}[htbp]
\centering
\includegraphics[scale=0.41]{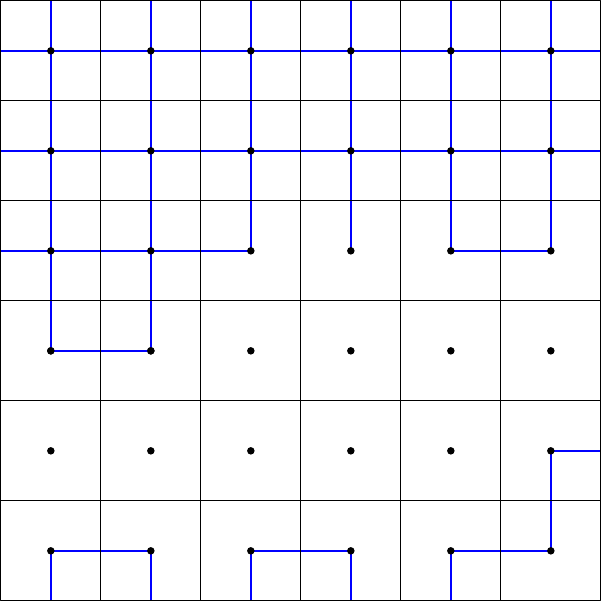}
\hspace{2ex}
\includegraphics[scale=0.41]{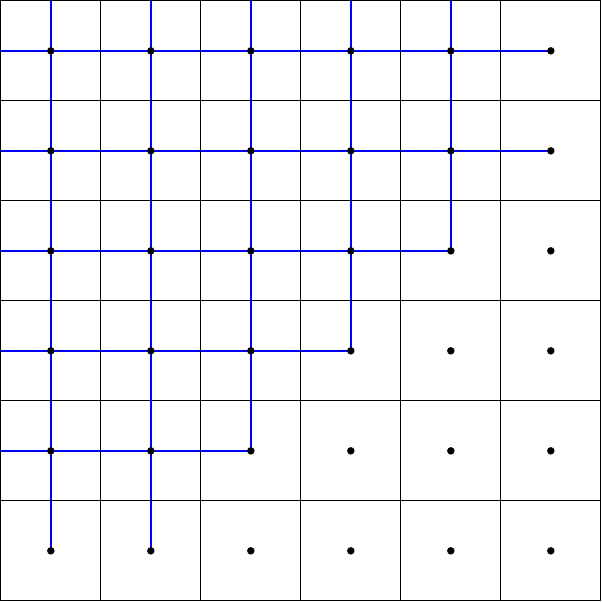}
\hspace{2ex}
\includegraphics[scale=0.41]{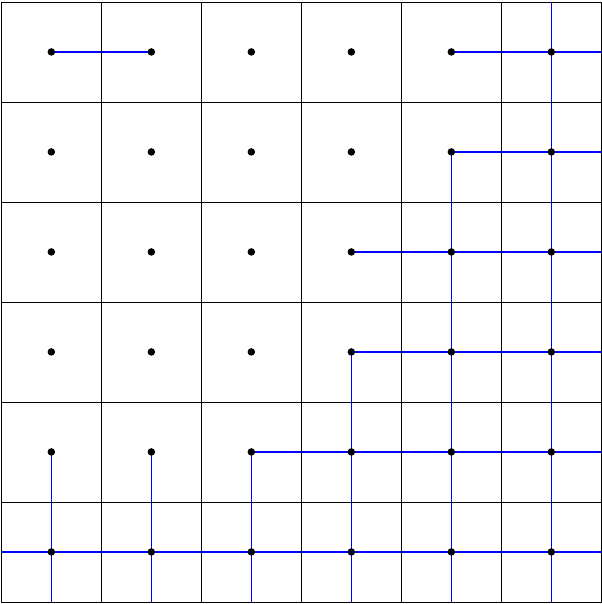}
\caption{Examples of tilings for  $\langle \I, \II, \IV, \VI \rangle$ corresponding to cases 1-3 (left to right). The middle and right should be seen as corners of a bigger grid filled with type $\I$ and $\VI$ tiles respectively.}
\label{fig:1246}
\end{figure}

We can prove that cases (1) and (2) can be solved in polynomial time;
they essentially can be solved by simple surface-area minimization arguments, and so can case (3) if  $n_4 > 1$.
But case (3) with $n_4 = 1$ is more challenging and again connects to Problem~\ref{prob:mcd}.
In this case, we have to delimit the area containing $\I$-tiles from $\VI$-tiles using only the $n_2$ $\II$-tiles and the single $\IV$-tile, which is the only tile that allows us to create a non-convex corner in a connected subregion containing $\VI$-tiles.
This enforces a single rectangle in the corner of our region containing all $\I$-tiles, with the sole $\IV$-tile in the corner of this rectangle and $\II$-tiles along the rest of the edges that are not adjacent to the grid border.
However, it is possible, and in some cases necessary, to use $\II$-tiles to create further rectangles containing $\VI$-tiles inside this rectangle, e.g., to avoid infeasible situations where $n_1 + n_2 + n_4$ is prime and would thus enforce a rectangle of height or width one which cannot contain a $\I$-tile;
for this purpose, we can again deploy a similar DP to the one used in \cref{thm:126}, checking all possible rectangle sizes (for which there are $O(wh)$ options).
Note that all $\II$-tiles also have to be in the separated area: since $n_2 < w$, we cannot make full columns or rows of $\II$-tiles and thus cannot place them in the region containing $\VI$-tiles.
\end{proof}

\section {Some Observations on Open Classes\label{sec:open}}
%\carola{Could consider moving this to the appendix and only summarize it in the Future Work section. }
%\carola{The Venn diagram figure is NOT updated -- there are 8 "?" but there should only be "6". Only include the figure in the appendix if it's updated.} \oswin{The Venn diagram IS and was updated! There are 6 LINES in the table with "?", but two of them stand for two classes, so in total there are 8 classes with "?" in the drawing. That also fits the 8 classes in this Section.}

 \begin {figure}[htbp]
   \centering
     \includegraphics [scale = 0.35, page=1] {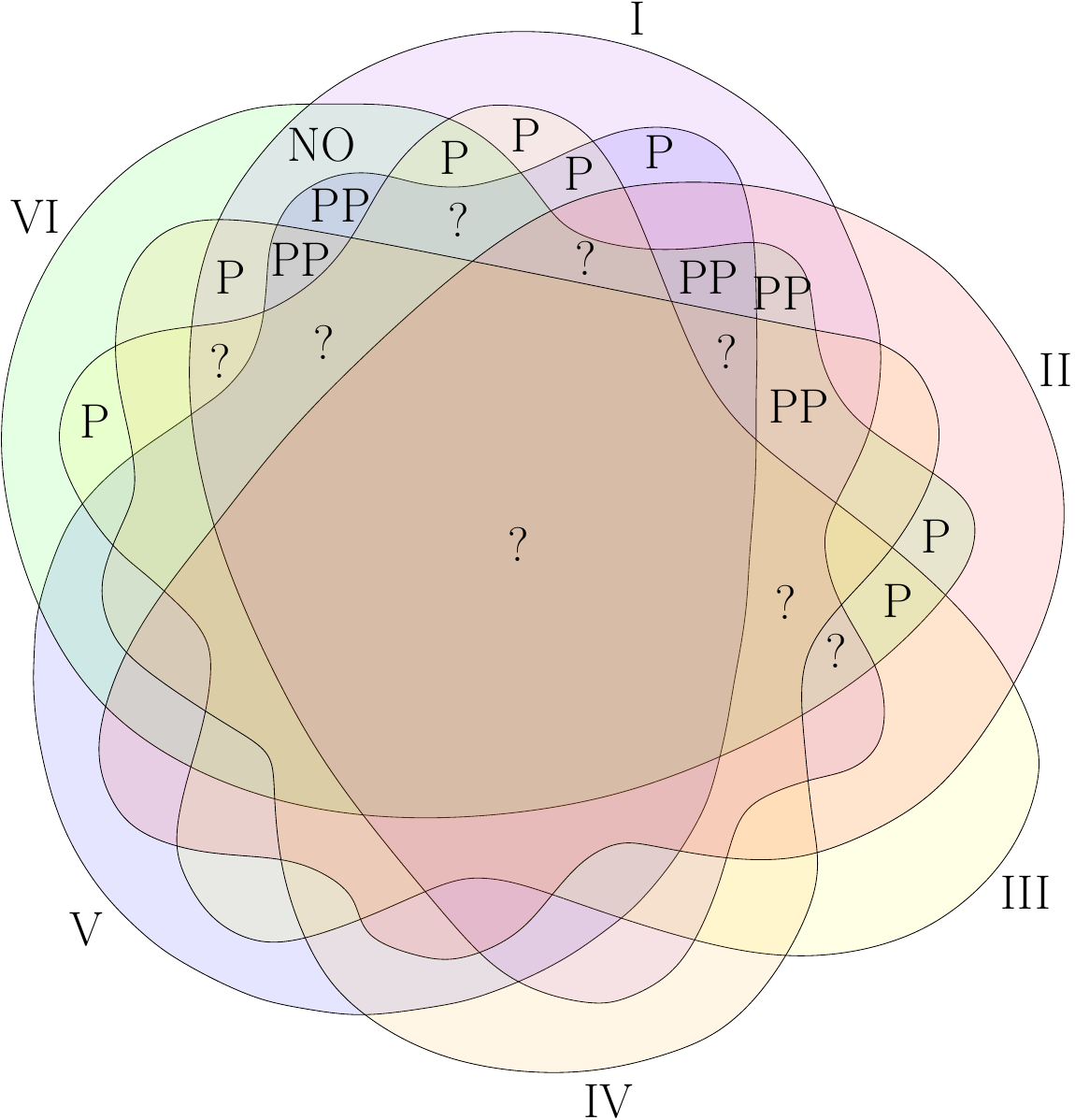}
   \caption {Venn diagram of all classes and their results as shown in Table~\ref{tab:results}, where PP stands for PseudoP. Classes without indicated result are YES instances. Observe that this is a full Venn diagram, that is, dual classes are not grouped together.%\carola{Check out the other pages of the figure in the figures folder}
   }
   \label {fig:venn}
 \end {figure}

While we have provided results for 33 of the 39 classes (counting dual classes as one), as summarized in Table~\ref{tab:results}, the status for six classes remains open. We conjecture that two of them are in \P{} and four are in \Q{}, including the class with all six tile types.

\subsection{Conjectured in \P}

{\bf $\langle \I, \III, \IV, \VI \rangle$.} We identified and solved many cases and our analysis suggests that an exhaustive characterization of all cases should lead to a polynomial-time algorithm. However, the number of cases appears to be very large and we have not characterized all of them.\\

\noindent
{\bf $\langle \I, \II, \IV, \V, \VI \rangle$.} A polynomial time algorithm may be surprising because this is again a class containing tiles of type $\TI, \TII, \TVI$ and it is similar to the other classes containing them. However, the situations that caused hardness for the similar classes cannot occur here.
When there is a yes instance it is always the case that all type $\TI$ and $\TVI$ tiles are clumped in separate clumps with a barrier between them made up of type $\TII$, $\TIV$, and $\TV$ tiles. Since there is at least one of each of these three tiles these clumps are never forced to be rectangular.

\subsection{Conjectured in \Q}
{\bf $\langle \I, \II, \III, \VI \rangle \hat= \langle \I, \III, \V, \VI \rangle$.}
We conjecture that a pseudo-polynomial algorithm similar to the classes in Section~\ref{sec:pseudoP} exists. Tiles of types $\TI, \TII, \TVI$ are also contained in this class and similar to $\langle \TI, \TII, \TVI \rangle$, there need to be rectangles formed with the type $\TVI$ tiles. However, the $ \TIII$-tiles can now also be used to surround these rectangles if they are in a particular location. \\

\noindent
{\bf $\langle \I, \II, \III, \IV, \VI \rangle \hat= \langle \I, \III, \IV, \V, \VI \rangle$.}
This class is similar to  $\langle \TI, \TII, \TIV, \TVI \rangle$, but now also type $\TIII$ tiles are available. We do not expect the type $\TIII$ to change the complexity of the problem and thus conjecture this class is of the same complexity as $\langle \TI, \TII, \TIV, \TVI \rangle$.\\

\noindent
{\bf $\langle \I, \II, \III, \V, \VI \rangle$.}
This class is similar to $\langle \TI, \TII, \TV, \TVI \rangle$, but now also type $\TIII$ tiles are available. Again we do not expect the type $\TIII$ to change the complexity of the problem and thus conjecture this class is of the same complexity as  $\langle \TI, \TII, \TIV, \TVI \rangle$.\\

\noindent
{\bf $\langle \I, \II, \III, \IV, \V, \VI \rangle$.}
This class, where at least one tile of each tile type exists, appears to be at least as hard as the other pseudopolynomial classes: it is unclear how the presence of type $\III$ tiles would make the problem easier.

\section {Future Work}
We have classified the complexity of deciding whether a perfect rectangular tiling exists for classes of subsets of the six tile types, leaving only six out of the 39 classes as open. A natural next step is to verify our conjectured results for these open cases.

Since the tiles can be considered as parts of a graph, each representing a vertex with up to four half-edges, a next natural step would be to investigate graph properties that may be achieved for specific instances, such as obtaining a connected graph, or a tree, or minimizing the area of the largest face / maximizing the area of the smallest face, etc.

Another interesting direction is to impose additional constraints at the boundary. Alternatively, one could eliminate the boundary completely and investigate a perfect tiling of the torus with the given graph tiles. % -- which could also be seen as having a boundary where the top and bottom as well as the left and right coincide.

\begin{comment}
We have characterized which multiplicities of at most three tile types are compatible. The clear next step is to extend this classification to classes with up to six different tile types.

Once this is understood, we can begin to investigate graph properties that may be achieved for specific instances, such as obtaining a connected graph, or a tree, or minimizing the area of the largest face / maximizing the area of the smallest face, etc.
\end{comment}

%\section{Acknowledgments}
%\carola{Add acknowledgments}

%%
%% Bibliography
%%
\bibliography{refs}

\clearpage
\appendix

\section{Hardness for Non-Rectangular Simple Polygons}
\begin{theorem}
   If tile counts are encoded in binary and the domain is given
   as a rectilinear polygon with coordinates in binary encoding, the problem is already \NP-hard if we only have tiles of types {\I} and \III.
   \label {thm:hardness-simple-polygons}
\end{theorem}
\begin{proof}
If we have binary encoded tile counts, we can use a very simple domain to reduce from \textsc{Subset Sum}.
Let $S_1, \ldots, S_k > 0$ be the input numbers for \textsc{Subset Sum} with $S = \sum_{i=1}^k S_i$ and desired sum $Z < S$.
W.l.o.g., assume that $S_1 \leq S_2 \leq \cdots \leq S_k$.
We only use tile types \I{} and \III{} in our reduction;
exactly $Z$ tiles are of type \III{}.
We construct a simple orthogonal polygon (actually, a histogram) row-by-row, starting from the top.
The first row has width $S_1$, the second row has width $S_2$, and so on;
after the row corresponding to $S_k$, we add $Z+1$ additional rows of width $Z+1$;
encoded as an orthogonal polygon, this yields a total of $O(k)$ corners with polynomial-size integer coordinates.

Firstly, we observe that due to the sorting and the $Z+1$ extra rows,
it is impossible to place any $\III$-tile vertically: for that, we would
have to make a full column, but all columns have height at least $Z+1$.
We also cannot place any $\III$-tiles in the extra rows, as these are $Z+1$ cells wide.
Thus, essentially, for each $S_i$, we allow the placement of exactly $S_i$ many type \III{} tiles in a specific row.
Therefore, if $Z$ can be expressed as a sum of the $S_i$ in which each $S_i$ is only used once, we can find a tiling; otherwise, it is impossible.
\end{proof}

\section {Additional Theorems and Proofs\label{app:yes}}

\subsection{YES and NO Classes}

\thmyesThree*
%\begin{theorem}
%  For classes
%  $\langle \I, \II, \III \rangle \hat= \langle \III, \V, \VI \rangle$,
%  $\langle \I, \II, \IV \rangle \hat= \langle \IV, \V, \VI \rangle$,
%  $\langle \I, \II, \V \rangle \hat= \langle \II, \V, \VI \rangle$,
%  $\langle \I, \III, \IV \rangle \hat= \langle \III, \IV, \VI \rangle$,
%  $\langle \I, \IV, \V \rangle \hat= \langle \II, \IV, \VI \rangle$,
%  $\langle \II, \III, \IV \rangle \hat= \langle \III, \IV, \V \rangle$,
%  $\langle \II, \III, \V \rangle$, $\langle \II, \IV, \V \rangle$,
%  a tiling always exists.
%\label{thm:cases3}
%\end{theorem}

\begin{figure}[htbp]
\centering
\includegraphics[width=.9\textwidth]{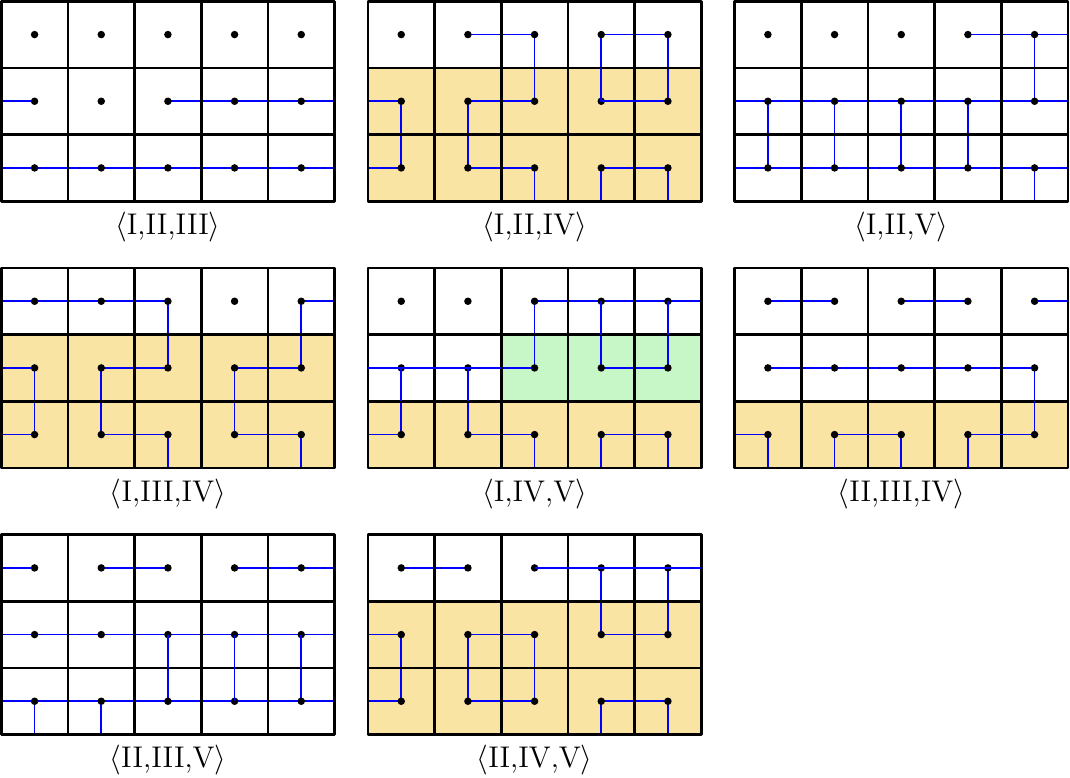}
\caption{Depiction of the tilings constructed in the proof of \cref{thm:cases3}; the yellow and green shaded areas each depict a single application of \cref{lem:type4-constrained-rect}.}
\label{fig:yes_cases3}
\end{figure}
\begin{proof}
For class $\langle \TI, \TII, \TIII \rangle$, simply place the type $\I$ tiles first and continue with type $\II$ tiles, filling rows in reading order until we run out.
If we do not end on a full row or with a halfedge pointing to the right, we can always modify the partially filled row by making its leftmost tile a type \II{} tile pointing towards the boundary, shifting the other tiles to the right.
We can then continue filling the remaining rows with type $\III$ tiles.

For class $\langle \TI, \TII, \TIV \rangle$, simply place the type $\I$ tiles first and continue with $\II$-tiles until they run out. Fill the remainder of a partially filled row with $\IV$-tiles, pointing their vertical halfedges downwards and alternating the orientation of their horizontal halfedges.
If unfilled rows remain, apply \cref{lem:type4-constrained-rect} to the remaining rectangle.

For class $\langle \TI, \TII, \TV \rangle$, simply place the type $\I$ tiles first and continue with type $\II$ tiles until they run out, making sure that we end with a full row or a half-edge pointing right (by starting the last row appropriately at the boundary).
Then fill the row with type $\V$ tiles with no halfedge pointing up. 
Then construct full rows of $\V$-tiles; the row above determines, for each cell individually, whether there is a halfedge pointing up and none pointing down or vice versa.

For class $\langle \TI, \TIII, \TIV \rangle$, let $\nI = kw + r$ for $0 \leq r < w$ and $\nIII = \ell w + s$ for $0 \leq s < w$.
First, we fill $k + \ell$ full rows of type $\I$ and $\III$ tiles,
leaving us with a $w \times (h - k - \ell)$ rectangle to fill with $r$ \I-tiles and $s$ \III-tiles alongside $\nIV$ \IV-tiles.
If $r + s < w$, we create a single row, starting with the $s$ tiles of type $\III$, followed by a single type $\IV$, then the $r$ tiles of type $\I$, and filled with as many more type $\IV$ as needed to fill the row;
we can then apply \cref{lem:type4-constrained-rect} on the remaining rows.
If $r + s \geq w$, due to $r + s \leq 2w - 2$ and $\nIV > 0$, we have $\nIV \geq 2$.
We can thus create two rows: one filled with $r$ \I-tiles, followed by a single $\IV$-tile in column $r+1$ and filled with \III-tiles, the other also with a single $\IV$-tile in column $r+1$ and the remaining $\III$-tiles left of that column; this row is then filled with $\IV$-tiles.
Afterwards, if there are remaining rows, we can again use \cref{lem:type4-constrained-rect} on them.

For class $\langle \TI, \TIV, \TV \rangle$, we first place all $\I$-tiles in reading order until they run out.
Then we place a single $\IV$-tile, followed by $\V$-tiles with no halfedge pointing up.
We then continue by filling rows with $\V$-tiles, with no halfedges pointing up or down depending on the row above, until they run out.
We then apply \cref{lem:type4-constrained-rect} on the remainder of the row the $\V$-tiles ran out in (if non-empty), and then again on the remaining rows (if any); see \cref{fig:yes_cases3}.

For class $\langle \TII, \TIII, \TIV \rangle$, we simply place all $\II$-tiles until they run out, making sure that we end with a halfedge pointing right.
We then continue filling rows with $\III$-tiles, filling any partially filled row with $\IV$-tiles before applying \cref{lem:type4-constrained-rect} on the remaining rows.

For class $\langle \TII, \TIII, \TV \rangle$, we place the type $\II$ tiles until they run out, making sure that we end with a half-edge (by starting appropriately at the boundary). Then we place type $\III$ tiles, followed by filling the row with type $\V$ tiles with a half-edge down.
The remaining rows can be filled with type $\V$-tiles (with no halfedge pointing up or down, depending on the row above).

Finally, for class $\langle \TII, \TIV, \TV \rangle$, we again first place the $\II$-tiles, making sure that we end with a halfedge pointing right.
We then fill rows with $\V$-tiles until they run out, fill the remaining partially-filled row with $\IV$-tiles and apply \cref{lem:type4-constrained-rect} on the remaining rows.
\end{proof}
%\subsection{YES Classes with four or five tile types}

\thmyesFourFive*

\begin{figure}[htbp]
\centering
\includegraphics[width=.9\textwidth]{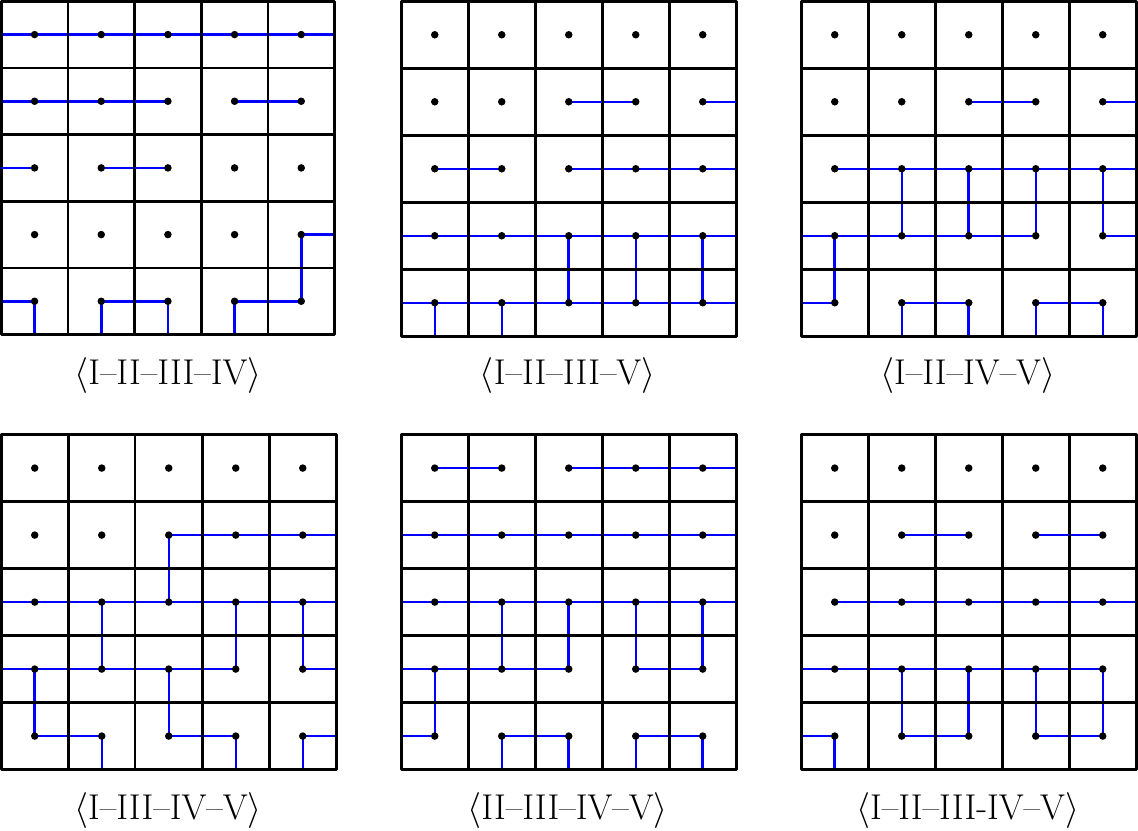}
\caption{Examples of the tilings constructed in the proof of \cref{thm:cases4-5}. }
\label{fig:yes_cases4-5}
\end{figure}

\begin{proof}

    For class $\langle \TI, \TII, \TIII, \TIV \rangle$, we first place the type $\III$ tiles until we run out. Then we place type $\II$ tiles until we run out, making sure to end without a half edge by optionally placing a single type $\II$ tile at the top left. Next we can place the type $\I$ tiles until we run out. Finally, we can place type $\IV$ tiles in the remaining space using \cref{lem:type4-constrained-rect}.
     
  For class $\langle \TI, \TII, \TIII, \TV \rangle$, we place the type $\I$ tiles until we run out. Then we place type $\II$ tiles until we run out, making sure to end with a half edge by optionally placing a single type $\II$ tile at the top left. Next we can place the type $\III$ tiles until we run out. Finally, we can place type $\V$ tiles in the remaining space, making sure to line up the sides without a half edge properly.

For class $\langle \TI, \TII, \TIV, \TV \rangle$, we place the type $\I$ tiles until we run out. Then we place type $\II$ tiles until we run out, making sure to end with a half edge by optionally placing a single type $\II$ tile at the top left. Next, we can place the type $\V$ tiles until we run out. Finally, we can place type $\IV$ tiles in the remaining space using \cref{lem:type4-constrained-rect}.

  For class $\langle \TI, \TIII, \TIV, \TV \rangle$, we place the type $\I$ tiles until we run out. Then we place a single type $\IV$ tile, with a single type $\V$ tile below it. If $\nIII + \nV < w-1$, place this type $\V$ tile with the side without a half edge to the left, otherwise this side goes down. Continue to the right of the type $\IV$ tile with the type $\III$ tiles until we run out. Finally place the type $\IV$ tiles using \cref{lem:type4-constrained-rect}.

For class $\langle \TII, \TIII, \TIV, \TV \rangle$, we place the type $\II$ tiles until we run out, making sure to end at a half edge by starting appropriately at the boundary. Then we place type $\III$ tiles until we run out. Next we can place the type $\V$ tiles until we run out. Finally, we can place type $\IV$ tiles in the remaining space using \cref{lem:type4-constrained-rect}.

For class $\langle \TI, \TII, \TIII, \TIV, \TV \rangle$, we place the type $\I$ tiles until we run out. Then place type $\II$ tiles until we run out, making sure to end with a half edge by optionally placing a single type $\II$ tile at the top left. Afterwards, we place the type $\III$ tiles until we run out. Next we can place the type $\V$ tiles until we run out. Finally, we can place type $\IV$ tiles in the remaining space using \cref{lem:type4-constrained-rect}.
\end{proof}

%%%%%%
\subsection{Polynomial Classes}

%\subsubsection{\texorpdfstring{$\langle \TI, \TIII, \TV \rangle \hat= \langle \TII, \TIII, \TVI \rangle$}{<I,III,V>=<II,III,VI>}}
\thmIiiiV*
\begin{proof}
    The proof follows the cases.
    It is easy to see that a tiling is impossible iff $\nIII + \nV < w \leq h$, since this does not allow us to make a single full row or column with these tiles; this is necessary for a valid tiling.

    If $\nIII + \nV$ is evenly divided by $w$ or $h$, we can make full rows or columns of either type of tile in addition to a single mixed row (or column), which can be placed at the boundary.

    If $\nV = 1$ and $s < w - q$, we can make a block of $q$ full columns of the \III-tiles with a distance of exactly $s+1$ to the left boundary.
    In the first column of that block, we can replace any $\III$-tile by the single $\V$-tile pointing left, placing the remaining $s + 1$ $\III$-tiles, including the replaced $\III$-tile, as a straight line connecting 
    the $\V$-tile and the left boundary.
    The case of $\nV = 1$ and $r < h - k$ works analogously using rows instead of columns.
    
    Otherwise, if $\nV = 1$, we cannot create a tiling.
    We can only make a single crossing using our single $\V$-tile.
    In the following, we argue the case that the $\V$-tile is positioned so that no halfedge points to the right; the other cases are analogous.
    Let $x \geq 0$ be the number of columns between the $\V$-tile and the left boundary.
    We thus know that $x$ \III-tiles must be placed between the $\V$-tile and the left boundary, and that column $C$ containing the $\V$-tile contains a total of $h-1$ $\III$-tiles.
    All other tiles to the left of $C$ must be $\I$-tiles; all remaining $\III$-tiles must be placed to the right of $C$ in full columns.
    Thus, $x$ must be chosen such that $h \mid (\nIII - h - x + 1)$.
    We show that, under the assumption that proposition (4) does not hold, 
    there are too many $\III$-tiles to fit right of $C$.
    W.l.o.g., we may therefore assume that the available space is maximized by minimizing $x$, i.e., setting $x = s + 1$; otherwise, we would be losing a full column of space for each increment of $x$.
    In that case, accounting for $C$ which contains $h-1$ $\III$-tiles,
    a total of $q-1$ full columns have to fit right of $C$;
    in other words, $w \geq x + 1 + (q-1) = s + 1 + q \Leftrightarrow s < w - q$ has to hold, contradicting our assumption.

    We have thus handled the case $\nV = 1$; 
    in the following, we assume $\nV \geq 2$.
    First, if none of the conditions (1)--(4) hold, we must have $\nV < 2(w-\nIII)$ and thus also $\nIII < w$.
    We have at least $w+1$ and at most $2w-2$ non-\I-tiles, of which at least $2$ are $\V$-tiles.
    
    We first observe that any tiling consisting solely of $\I, \III$ and $\V$-tiles with $\nIII + \nV < 2w-1 \leq 2h-1$ must have a straight-line path connecting one boundary with the opposite one.
    In particular, we do not have enough tiles to simultaneously connect the left and right and the top and bottom boundaries, or to completely fill two columns or rows.
    Any straight-line path that starts at one boundary must terminate at the opposite one, unless it is stopped by a $\V$-tile.
    This tile then starts a perpendicular straight-line path, which must itself be stopped by another $\V$-tile and so on, until we are either left with a straight-line path connecting two opposite boundaries or a cycle with antennas connecting it to each boundary; the latter requires at least $2w$ tiles and can thus be excluded in this case.
    Because $\nIII < w$, this straight-line path must contain at least $w-\nIII$ $\V$-tiles, and all non-$\I$-tiles must lie on one side of that path.
    This implies that this path does not lie adjacent to a boundary.
    Therefore, any $\V$-tile on this path needs to have another tile next to it.
    This contradicts $\nV < 2(w - \nIII)$.

To analyze the remaining case that only proposition (3) holds, in which we claim that a tiling always exists, we first consider the case $\nIII < w$,
split into three sub-cases (a)~$\nIII + \nV \leq 2w$, (b)~$2w < \nIII + \nV < 3w$ and (c)~$3w \leq \nIII + \nV$.
For (a), we can achieve a tiling that places all $\III$-tiles in the second-lowest row as depicted in \cref{fig:135_case_a}.
    \begin{figure}[htbp]
        \centering
        \includegraphics[width=.7\textwidth]{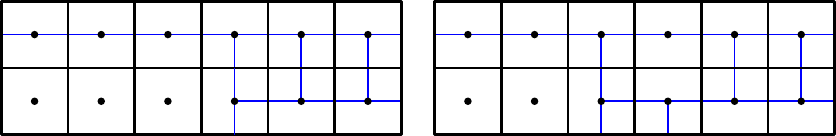}
        \caption{Handling case (3), subcase (a): $\nIII + \nV \leq 2w$ and $\nV \geq 2(w-\nIII)$. Left: handling $\nV = 2(w-\nIII)$; right: increasing $\nV$ by one.}
        \label{fig:135_case_a}
    \end{figure}

    In subcase (c)~$\nIII + \nV \geq 3w$, we can place all $\III$-tiles in the first row, fill it with $\V$-tiles and handle the remaining $\V$-tiles analogously to our construction for the class $\langle \TI, \TV \rangle$,
    which only requires three boundaries to be accessible.
    For (b)~$2w < \nIII + \nV < 3w$, we can create a tiling as depicted in \cref{fig:limited-range-2n-3n}.

    Finally, we consider the case $\nIII \geq w$ and $\nV \geq 2$.
    In this case, we begin by filling complete rows at the bottom boundary
    with $\V$-tiles until the remaining number $\nV$ of $\V$-tiles drops below $w$.
    If $\nV = 0$ tiles remain, we replace two $\V$-tiles in the filled rows by
    $\III$-tiles; depending on whether we filled an even or an odd number of rows,
    we can use either the two leftmost tiles in the lowest row or
    the leftmost tile in the two lowest rows.
    Similarly, if $\nV = 1$ tile remains, we replace one $\V$-tile in the filled rows by a $\III$-tile.
    This ensures that $\nV \in [2, w-1]$ $\V$-tiles remain.
    We can furthermore guarantee that $\nIII \geq w-1$ $\III$-tiles remain;
    in the only case in which we end up with $w-2$ $\III$-tiles, $\nIII = w$ and $w \mid \nV$, and we would not be in this case.
    We put aside $2$ $\V$-tiles and $w-1$ $\III$-tiles.
    We then construct full rows of $\III$-tiles above the block of 
    complete rows of $\V$-tiles until what remains does not suffice for a full row.
    If, not including the tiles put aside, the remaining tiles together can 
    fill another row, we construct such a row at the top boundary;
    in that case, we put all remaining $\V$-tiles, excluding the two tiles put aside, into that row.
    Note that this row must contain at least one $\V$-tile.
    We then place the remaining tiles, including the tiles put aside, as depicted in \cref{fig:i-iii-v-put-aside}.
\end{proof}

    \begin{figure}[htbp]
        \centering
        \includegraphics[width=.3\textwidth, angle=90]{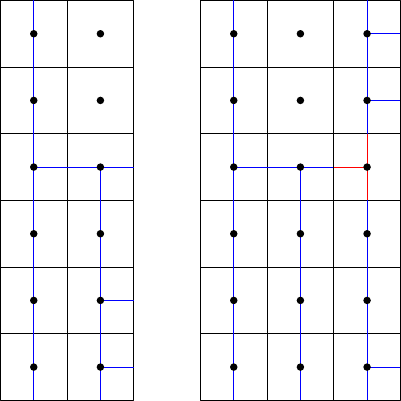}
        \caption{Using at least $w-1$ $\III$-tiles and at least $2$ $\V$-tiles, we can accommodate the remaining collection of $\III$- and $\V$-tiles. \textbf{Bottom:}~not including the tiles put aside, fewer than $w$ tiles remained, and we did not create a mixed $\III$- and $\V$-row at the top boundary. \textbf{Top:}~we created an additional full row at the boundary, ensuring that the only tile that needs an incoming half-edge on its bottom side (red) is a suitably rotated $\V$-tile. In this case, the two rows below the topmost one only contain $\III$-tiles and the two $\V$-tiles previously put aside.}
        \label{fig:i-iii-v-put-aside}
    \end{figure}%

    \begin{figure}[htbp]
        \centering
        \includegraphics[width=.6\linewidth]{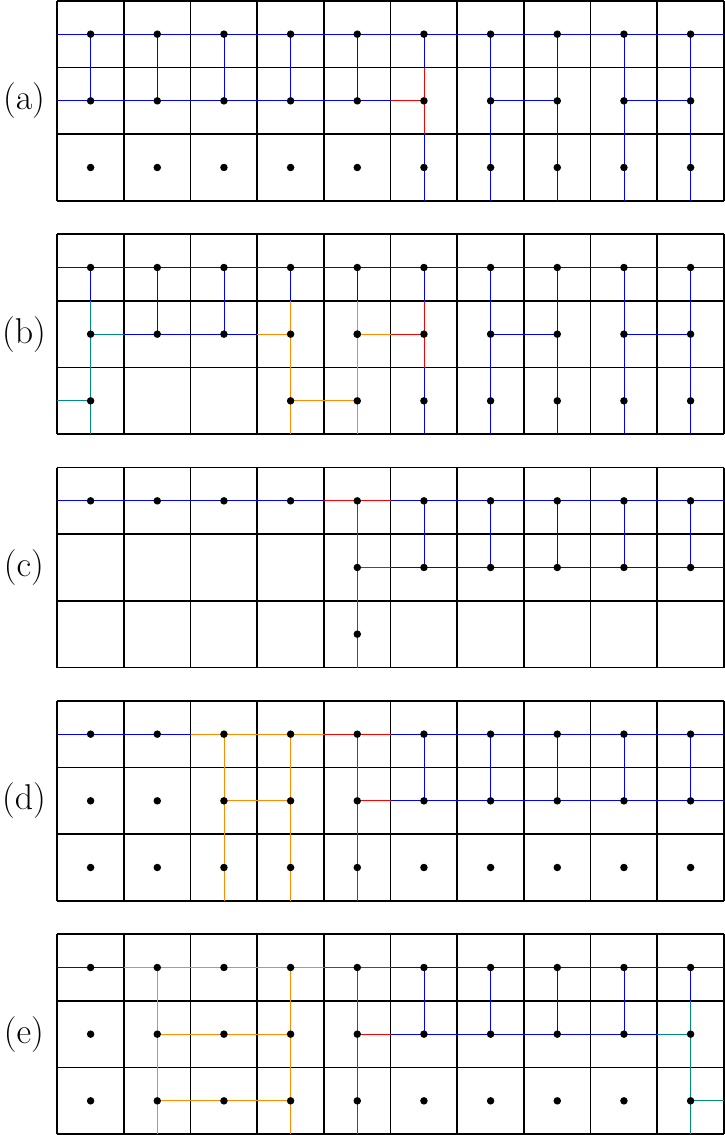}
        \caption{Creating a tiling with $\nIII < w$ and 
                 $2w < \nIII + \nV < 3w$.
                 (a,b)~With at least $2w$ $\V$-tiles, we can fill two
                 rows with a ladder-like structure, placing all $\III$-tiles
                 pointing towards the boundary.
                 By using the right boundary, we resolve parity issues to ensure the red $\V$-tile
                 points to the right. (b)~shows how we can add two (orange) or one (cyan) additional $\V$-tiles to the base construction. (c,d,e)~With fewer than $2w$ $\V$-tiles,
                 we fill the third row from the bottom using the basic construction shown in~(c). (d)~shows how to add four additional $\V$-tiles (orange); due to the bound on $\nV$, we cannot run out of space during this construction.
                 (e)~shows how to modify
                 an addition of four tiles to instead add six tiles (orange);
                 together with adding one tile (cyan), this allows us to
                 handle moduli $1$ (cyan only), $2$ (orange only) and $3$ (cyan and orange) modulo $4$, and thus any number of additional $\V$-tiles.}
        \label{fig:limited-range-2n-3n}
    \end{figure}
%This lemma, together with the fact that for $n < 4$,
%the existence of a tiling can simply be checked via brute force,
%shows that finding a tiling if one exists is in \P{} in this case.
%\oswin{Some incomplete thoughts: Let the number of tiles of type $\I$ be 
% $N_\I = kn+r$ for $0 \le r < n$, 
%  and the number of tiles of type $\III$ be $N_\III$ and of type $\V$ be
% $N_\V$.
%  If $r=0$ then any combination of tiles of type $\III$ and $\V$ can fill a block of size $(n-k) \times n$ as previously explained. So assume $r>0$.  If $N_\V=1$ then it is easy to see that we find a tiling if and only if $n-1 \leq N_\III \leq 2n-2$. If $N_\I$ is too large, that is, $k \geq n-1$, then no tiling is possible. So assume $k \leq n-2$. If $N_\V = 2$ then we can make a tiling as shown in the figure on the right. IF $N_\V > 2$ we can replace some tiles of type $\I$ by tiles of type $\V$, but if $N_\V$ gets too large this needs additional arguments and/or a slightly different construction. }

%\vspace{1em} \noindent {\bf value:} \P
%%%%%%%%%%%%%%%%%%%%%%%
\subsection{Pseudopolynomial Classes}
Furthermore, for any fixed $w, h$ and $\nVI$, replacing $\I$-tiles by $\II$-tiles never hurts.
\lemmintwoopt*
\begin{proof}
    Recall we assume $w \leq h$.
    We show that we can replace a single $\I$-tile by a $\II$-tile,
    i.e., we assume $\nII' = \nII+1$; the lemma follows by induction.
    If $\nII' \geq 2w$, the lemma holds by \cref{lem:2n-2-tiles-suffice}.
    Thus, we can assume $\nII < 2w-1$.
    Consider a tiling $T$ for the original instance.
    If there is a $\I$-tile adjacent to a boundary,
    we can replace it by a $\II$-tile pointing towards the boundary and are done.
    We now want to show that, if we cannot find such a tile, there are already at least $2w-1$ $\II$-tiles.
    
    If there are complete rows of $\VI$-tiles on both the bottom and the top boundary (or complete columns on the left and right boundary), due to $\nII, \nI > 0$, we know that there must be at least $2w \leq 2h$ $\II$-tiles to separate an area between the two blocks of full rows/columns.
    
    If there are complete rows or columns of $\VI$-tiles on exactly one boundary, say, w.l.o.g., rows at the top boundary, then the only case in which there are not two $\II$-tiles in each column is if $h-1$ rows are filled with $\VI$-tiles; but that contradicts $\nI > 0$.

    In the case that there are neither full rows nor full columns of $\VI$-tiles,
    the only way in which a column may not contain at least $2$ $\II$-tiles is if it contains exactly one (either at the bottom or at the top) and is otherwise filled with $\VI$-tiles.
    In that case, there is a rectangle of height $h-1$ and some width $x$ filled with $\VI$-tiles; there cannot be more than one such rectangle since each contributes at least $h \geq w$ $\II$-tiles.
    This rectangle occupies $x+1$ or $x+2$ columns (including its \II-tile boundary); in the latter case, its boundary already contains $2w-1$ $\II$-tiles.
    Otherwise, to avoid a $\I$-tile on the region boundary, its boundary contains $x + h$ tiles; if $x+h \geq 2x + 2 \Leftrightarrow h \geq x + 2$, the columns of this rectangle also contain at least $2$ $\II$-tiles on average.
    Due to $h \geq w \geq x+1$, the only remaining case is $h = w = x+1$, which contradicts $\nI > 0$.
\end{proof}

\subsubsection{\texorpdfstring{$\langle \TI, \TII, \TIV, \TVI \rangle \hat= \langle \TI, \TIV, \TV, \TVI \rangle$}
{<I,II,IV,VI>=<I,IV,V,TVI>}}
\ppAndRSAHardness*
\begin{proof}
We have to separate the $\I$-tiles and $\VI$-tiles by $\II$- and $\IV$-tiles, since $\I$- and $\VI$-tiles cannot be adjacent.
\begin{figure}[htbp]
\centering
\includegraphics[scale=0.41]{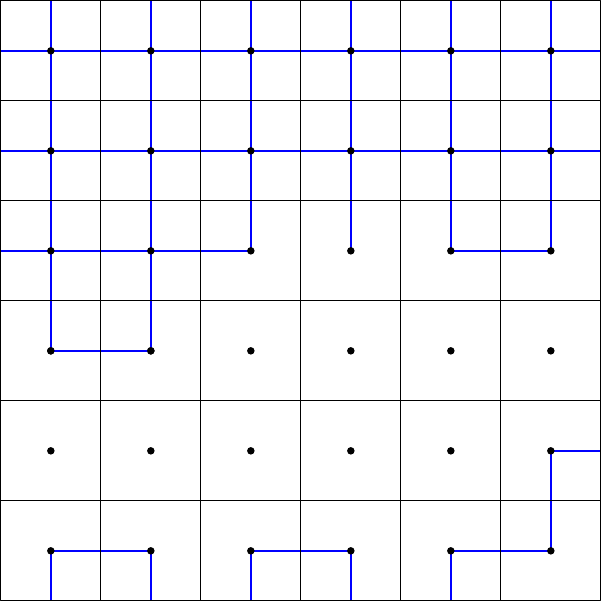}
\hspace{3ex}
\includegraphics[scale=0.41]{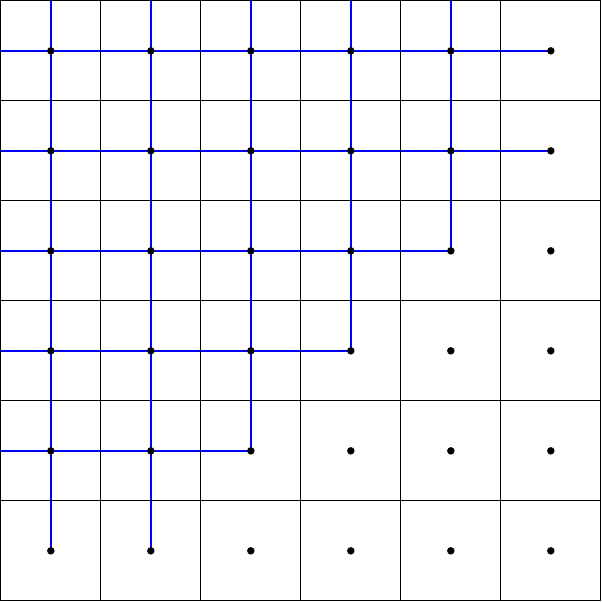}
%\hspace{1ex}
\includegraphics[scale=0.41]{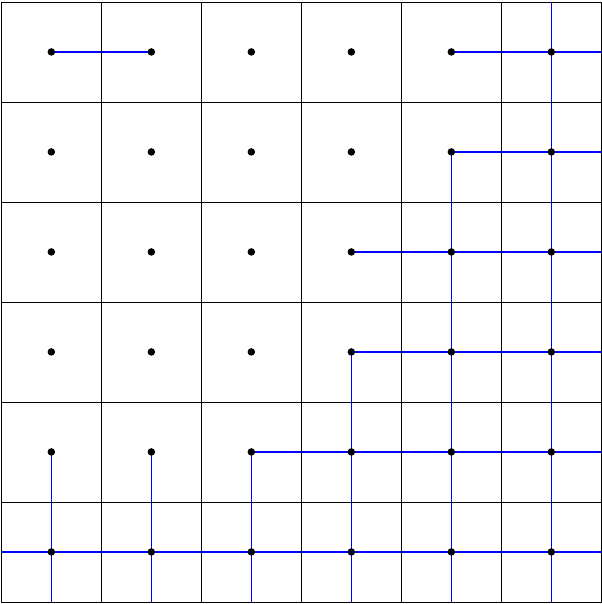}
\caption{Examples of tilings for  $\langle \I, \II, \IV, \VI \rangle$ corresponding to cases 1-3 (left to right). The middle and right should be seen as corners of a bigger grid filled with type $\I$ and $\VI$ tiles respectively.}
\label{fig:1246-a}
\end{figure}%
We distinguish the following cases, depicted in \cref{fig:1246-a}:
(1)~a horizontal or vertical path consisting of $\II$- and $\IV$-tiles connecting two opposite boundaries separates two blocks of $\I$- and $\VI$-tiles,
(2)~all $\VI$-tiles are placed in a region of the grid that contains one corner, separated by a path of $\II$- and $\IV$-tiles connecting two consecutive boundaries,
(3)~all $\I$-tiles are placed in a region of the grid that contains one corner, separated by a path of $\II$- and $\IV$-tiles connecting two consecutive boundaries.
We will argue later why the remaining cases, namely multiple disconnected regions of $\VI$- or $\I$-tiles separated by closed loops of $\II$- and $\IV$-tiles or regions on the boundary separated by a path of $\II$- or $\IV$-tiles can be transformed into tilings that fall into one of the cases (1)--(3) and thus, we can assume that our case distinction is exhaustive.
The idea of the argument for this is that, if we have sufficiently many $\II$- and $\IV$-tiles, we can compute tilings in which we pretend the excess $\II$- and $\IV$-tiles to be $\I$-tiles and then repair these by replacing some $\I$-tiles with $\II$- and $\IV$-tiles.
Intuitively speaking, this allows us to focus on minimizing the surface area where the regions containing $\I$-tiles and $\VI$-tiles touch, which is achieved by one of the cases (1)--(3).

{\bf Case (1).}
A tiling of type (1) is possible iff $n_4 + n_2 \geq w$.
All $\VI$-tiles can be placed at the top of the grid in reading order,
yielding at most one partially-filled row.
If a partial row is present, a single $\IV$-tile suffices to stop the $\VI$-tile block.
Afterwards, any empty tile below a $\VI$-tile can be filled by either $\IV$- or $\II$-tiles; any parity issues with the $\IV$-tiles can be solved at the boundary.
Next, any remaining $\II$-tiles can be placed in pairs, again repairing parity at the border, followed by all $\I$-tiles.
Finally, the remaining space can be filled with $\IV$-tiles using \cref{lem:type4-constrained-rect}.
Thus, in the following, we can assume $n_4 + n_2 < w$.

{\bf Case (2).}
%\begin{lemma}
%    Let $a = \lfloor n_2/2 \rfloor\lceil n_2/2\rceil + n_2n_4 + n_4(n_4-1)/2$.
%    A tiling of type (2) is possible iff $\nVI \leq a$.
%\end{lemma}
%\begin{proof}
Let $a = \lfloor n_2/2 \rfloor\lceil n_2/2\rceil + n_2n_4 + n_4(n_4-1)/2$. We  claim that a tiling of type (2) is possible iff $\nVI \leq a$. We prove this claim in the remainder of this case:
A tiling of type (2) contains a path of $\II$- and $\IV$-tiles connecting one side of the boundary to a consecutive side; w.l.o.g., let that be the left side and the top side as depicted in \cref{fig:1246}~(middle).
    We will later argue that, in this case, we can modify tilings for a reduced $\nVI$ if $\nII, \nIV$ stay constant; thus, we want to maximize the area of the region that can receive $\VI$-tiles, that is, that we can separate from the region containing $\I$-tiles.
    Moreover, we can assume that the whole region it separates is filled with $\VI$-tiles, as otherwise we can just make the separated region smaller.
    Within such a path, a $\II$-tile can either correspond to a horizontal or a vertical move, whereas a $\IV$-tile corresponds to a horizontal, vertical or diagonal move, with diagonal moves resulting in a larger area gain compared to only using horizontal and vertical moves.
    To maximize the separated area, we therefore construct a path from the left to the top boundary, starting in $\lfloor\nII/2\rfloor$ $\II$-tiles and ending in $\lceil\nII/2\rceil$ $\II$-tiles, connected by $\nIV$ diagonal steps using $\IV$-tiles; up to swapping the number of horizontal and vertical steps done via $\II$-tiles, this is the unique strategy that maximizes the separated area,
    and it separates a total area of $a$, which can be completely filled with $\VI$-tiles, allowing us a tiling with up to $a$ $\VI$-tiles.
    
    To handle any lower number of $\VI$-tiles, we can erase parts of separated columns from right to left by moving the rightmost $\IV$-tile upwards and swapping the \II-tile to its right into its place.
    Full columns or rows can also be removed.
    Any extra $\II$- and $\IV$-tiles can be placed in the region containing the $\I$-tiles.

{\bf Case (3), $n_4 > 1$.}
We subdivide Case (3) into two subcases, $n_4 > 1$ and $n_4 = 1$.
The case $n_4 > 1$ is similar to Case (2) with the $\I$-tiles confined to the separated corner.
A notable difference is that $\II$-tiles cannot be placed adjacent to multiple $\VI$-tiles like they are to $\I$-tiles in Case (2).
This means that the type $\IV$ tiles are placed one tile back, resulting in $n_2$ fewer $\I$-tiles being allowed than $\VI$-tiles in Case (2).
Finally, any excess $\II$-tiles must be placed inside the type $\I$ area; see \cref{fig:1246}~(right).

{\bf Case (3), $n_4 = 1$}
This is the difficult case that connects back to Problem~\ref{prob:mcd}.
We have to delimit the area containing $\I$-tiles from $\VI$-tiles using only the $n_2$ $\II$-tiles and the single $\IV$-tile.
Since we are not in case 1, we cannot do so by filling an entire row or column with these tiles.
That means the entire circumference of the area containing $\I$-tiles must be surrounded by the boundary, $\II$-tiles or $\IV$-tiles.
We have only one $\IV$-tile available; since this is the only tile that allows us to create a non-convex corner in a connected subregion containing $\VI$-tiles, this enforces a single rectangle $\mathcal{R}$ in a corner of our region containing all $\I$-tiles; w.l.o.g., let that be the top-left corner of our region.
The sole $\IV$-tile has to sit in the bottom-left corner of this rectangle and $\II$-tiles along the right and bottom edges of $\mathcal{R}$ separate the $\I$-tiles in $\mathcal{R}$ from the $\VI$-tiles outside.
Note that all $\II$-tiles also have to be in $\mathcal{R}$: since $n_2 < w$, we cannot make full columns or rows of $\II$-tiles and thus cannot place them in the region containing $\VI$-tiles.

However, this does not mean that $\mathcal{R}$ needs to be free of $\VI$-tiles: indeed, we can take $\II$-tiles to frame rectangles of $\VI$-tiles and place these framed rectangles inside $\mathcal{R}$, as long as they do not overlap and do not touch the right or bottom boundary of $\mathcal{R}$.
This approach is sometimes necessary: otherwise, e.g., if $n_1 + n_2 + n_4$ is prime, it would force $\mathcal{R}$ to be of width or height $1$, but such a rectangle cannot contain a $\I$-tile; adding some additional $\VI$-tiles into $\mathcal{R}$ increases the necessary area, but may decrease the circumference of $\mathcal{R}$ by allowing a more favorable factorization of its area.

Thus, unlike the situation for cases~(1)~and~(2) or case~(3) with $n_4 > 1$, where we have simple polynomial-time checkable formulas that describe whether a tiling corresponding to either case exists, for this case we only have a pseudo-polynomial time algorithm.
For each of the $O(wh)$ possible height/width-combinations $(w_\mathcal{R}, h_\mathcal{R})$ of $\mathcal{R}$,
this algorithm runs a slight variant of the DP used for \cref{thm:126} to find the minimum number of $\II$-tiles needed to place the $n_6' = w_\mathcal{R}h_\mathcal{R} - n_1 - n_2 - n_4$ necessary $\VI$-tiles in $\mathcal{R}$; if that number is at most $n_2 - w_\mathcal{R} - h_\mathcal{R} + 2$ for any possible $(w_\mathcal{R}, h_\mathcal{R})$, a tiling exists;
otherwise, no tiling is possible.

{\bf Sufficiency of cases.}
Let $T$ be any tiling for this class.
We describe how we can map it to a tiling that fits one of the Cases (1)--(3).
Firstly, if $n_2 + n_4 \geq w$, we already described how we can handle all multiplicities in a way that fits Case (1).
Therefore, we can restrict our attention to the case $n_2 + n_4 < w$;
consider a tiling $T$ for this case.
Since there are both $\I$-tiles and $\VI$-tiles that cannot be adjacent,
$T$ must contain at least one path from one boundary to another, or a loop of $\II$- and $\IV$-tiles separating $\I$- and $\VI$-tiles.
Such a path has a bounded region (the inside of the loop, or the inside of the shorter loop formed by a path and the boundary); if the bounded region contains $\I$-tiles, we call the path or loop a $\I$-separator; otherwise, it is called a $\VI$-separator.
If $n_4 = 1$ and there is a $\I$-separator, the tiling $T$ must already be according to case (3).
Otherwise, let $s_1$ and $s_6$ be the number of tiles separated by $\I$-separators and $\VI$-separators.
Since both cases~(2)~and~(3) with $n_4 > 1$ construct a $\I$- or $\VI$-separator that maximizes the separated area among all such separators and can handle any smaller number of separated tiles as well, we can simply collect either all $s_1$ tiles or all $s_6$ tiles separated in $T$ and instead separate them as they would be separated in case (2) (for $s_6$) or case (3) (for $s_1$).
In either case, we obtain a tiling that corresponds to one of the cases (1)--(3).
\end{proof}

\subsection{\texorpdfstring{$\langle \TI, \TII, \TV, \TVI \rangle$}
{<I,II,V,VI>}}

\begin{theorem}
For the class $\langle \TI, \TII, \TV, \TVI \rangle$, there is a pseudo-polynomial time algorithm deciding whether a tiling exists.
\end{theorem}
\begin{figure}[htbp]
\centering
\includegraphics[scale=0.41]{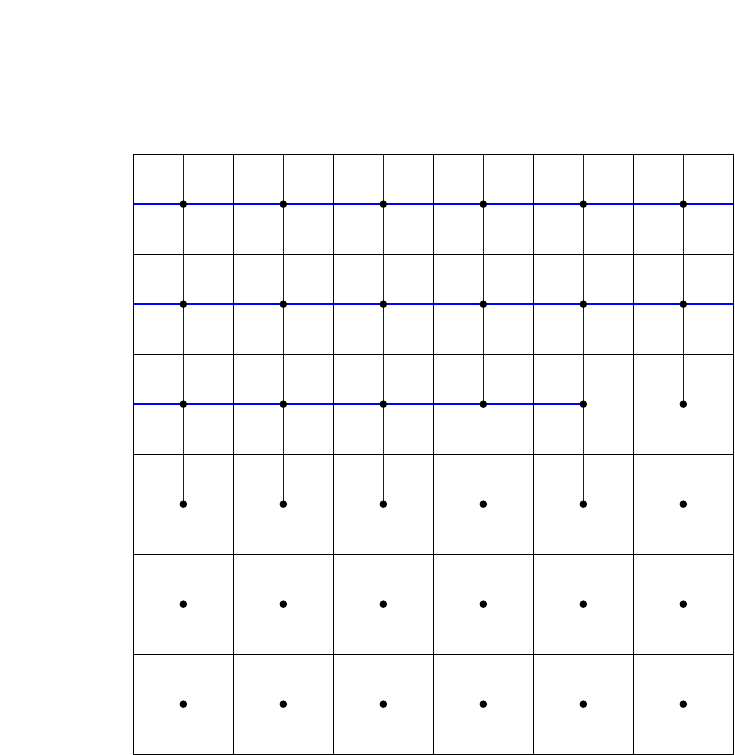}
\hspace{1cm}
\includegraphics[scale=0.41]{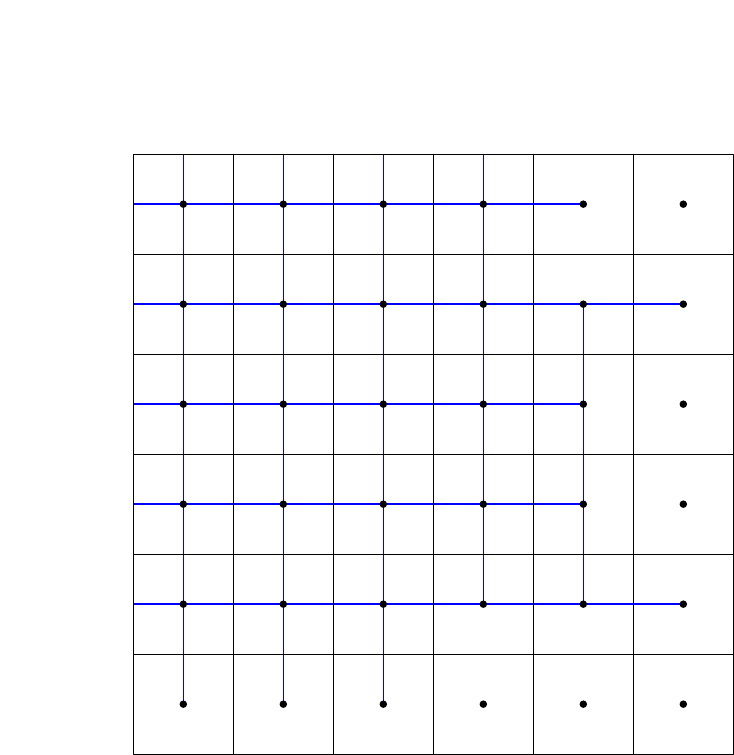}
\caption{Examples of tilings for the class $\langle \I, \II, \V, \VI \rangle$ corresponding to the cases 2 (left) and 5 (right). The right picture should be seen as a corner of a bigger grid filled with type $\I$ tiles.}
\end{figure}
\begin{proof}
We once again have to separate the  $\I$ and $\VI$ tiles since they can never be adjacent.
We distinguish the following cases where placement is always possible.
The cases should be read in an if-else fashion.
Additionally, keeping in mind that this set is its own dual, all cases also exist in a dual variant, which is not explicitly mentioned.\\
\textbf{Case (1)} $\displaystyle n_2 \equiv n_6 \mod{w} \land n_5 = w - n_2$ \\
Place full rows of type $\VI$ tiles until there are not enough left, then place a partial row of type $\VI$ tiles followed by all type $\V$ tiles, this should fill exactly one row. Next, place all type $\II$ tiles with their half-edge up on the next row. The remaining space can be filled with type $\I$ tiles. \\
\textbf{Case (2)} $\displaystyle n_2 + n_6 > w$ \\
In this case, we need $w + 1$ type $\II$ and $\V$ tiles to act as a barrier (including at least one of each).
Once these tiles have been picked the grid can be filled as follows. First, place the type $\V$ tiles that are not needed for the barrier, making sure that the lowest row has half-edges going out downward. Next, place all type $\VI$ tiles. Afterwards, place the type $\V$ tiles for the barrier with their side without a half-edge downward, the last type $\V$ tile has to be rotated to have the side without half-edge be to the right instead (assuming we are filling rows left to right). Next, a single type $\II$ tile can be placed below this last type $\V$ tile, then the rest of the type $\II$ tiles for the barrier can be placed to make sure all half-edges are completed. Then all type $\I$ tiles can be placed. Finally, the remaining type $\II$ can be placed in pairs. \\
\textbf{Case (3)} $\displaystyle \text{let } a = \frac{n_2 + n_5-1}{2}; n_6 \leq \lfloor a \rfloor\lceil a \rceil +\lceil a \rceil \land n_5 \leq \lceil a \rceil - (n_6 - \lfloor a \rfloor\lceil a \rceil)$ \\
If $n_6 < \lfloor a \rfloor\lceil a \rceil$ the structure from case 5 should be used instead. \\
In this case, $\lfloor a \rfloor\lceil a \rceil$ type $\VI$ tiles can be placed in a $\lfloor a \rfloor \times \lceil a \rceil$ rectangle in the top-left corner, with the rest of the type $\VI$ tiles in the row below it. Then the type $\V$ tiles can be placed with their side without a half-edge facing down, the last type $\V$ tile should be rotated to have the side without half-edge to the right instead. Next, all remaining dangling half-edges can be connected to a type $\II$ tile, with the remaining type $\II$ tiles placed on an opposite border of the grid. Finally, the rest of the area can be filled with the type $\I$ tiles. \\
\textbf{Case (4)} $\displaystyle \text{let } a = \frac{n_2 + n_5-1}{2}; n_6 \leq \lfloor a \rfloor\lceil a \rceil +\lceil a \rceil \land n_5 = \lfloor a \rfloor + \lceil a \rceil  - (n_6 - \lfloor a \rfloor\lceil a \rceil)$ \\
If $n_6 < \lfloor a \rfloor\lceil a \rceil$ the structure from case 5 should be used instead. \\
In this case, $\lfloor a \rfloor\lceil a \rceil$ type $\VI$ tiles can be placed in a $\lfloor a \rfloor \times \lceil a \rceil$ rectangle in the top-left corner, with the rest of the type $\VI$ tiles in the row below it. Then the type $\V$ tiles can be placed with their side without a half-edge facing down to fill the rest of the row. Then the column to the right of this rectangle should be filled with type $\V$ tiles with their side without a half-edge facing to the right. A single type $\II$ tile should be placed below this last type $\V$ tile. The other type $\II$ tiles can be placed on an opposite border of the grid. Finally, the rest of the area can be filled with the type $\I$ tiles. \\
\textbf{Case (5)} $\displaystyle \text{let } b = \frac{n_2 + n_5-3}{2}; n_6 \leq \lfloor b \rfloor\lceil b \rceil +\lceil b \rceil \land n_2 \geq 2$ \\
In this case, we first need to carefully select tiles to act as a barrier between the type $\I$ and $\VI$ tiles.
We need $\lceil b\rceil + 3 + c$ tiles, where $c$ is the amount of complete size $\lceil b \rceil$ rows we can make using type $\VI$ tiles and the $\V$ tiles not needed for the barrier.
At least 2 of the barrier tiles should be of type $\II$, and at least 1 of type $\V$. If the amount of type $\V$ tiles for the barrier + the amount of type $\VI$ and $\V$ tiles left over after making the $c$ rows is not enough to make another row, then one less tile is needed for the barrier.
We can now place the type~$\V$ tiles not needed for the barrier and the type $\VI$ tiles in the top-left corner in $c$ rows with one partial row below it, making sure that there are dangling half-edges to the outside of this block. We can now place the barrier type $\V$ to fill up the rest of this partial row, with their side without a half-edge facing down. If we run out of type $\V$ tiles before filling the row, the last one should be rotated to have the side without a half-edge facing right instead. Otherwise, we can then place another type $\V$ tile to the right of this row. Next, we can place type $\V$ tiles above this last one until we run out or reach the border of the grid, they should be placed with their side without a half-edge to the right. If we run out then the last one should be rotated to have the side without a half-edge to the top instead. If we did not run out, we can place the remaining type $\V$ tiles below the type $\VI$ tiles in our bottom row, their sides without a half edge should face down. The last one should be rotated to have that side face right instead. Any remaining dangling half-edges should be filled with the type $\II$ tiles we have reserved for the barrier. The rest of the type $\II$ tiles can be placed on an opposite border of the grid. Finally, the rest of the grid can be filled with type $\I$ tiles. 
\\ \\
Any other situation with at least two of each type $\II$ and $\V$ tiles is impossible to place, since there will not be enough tiles to create a barrier between the type $\I$ and $\VI$ tiles. \\
The only open situation is when there is exactly one type $\II \hat= \V$ tile. This situation results in a rectangular area in a corner of the grid where all the type $\VI$ and $\V$ tiles are placed. This area can have some holes containing type $\I$ tiles. This area should have type $\V$ tiles along all borders (internal or external). The single type $\II$ tile can be used in the corner of the area opposite to the corner of the grid. This is a similar structure as described for the $\langle \I,\II,\IV,\VI \rangle$ class. The same DP algorithm can be used. 
\end{proof}

\end{document}